\documentclass[times,doublespace]{amsart}
\usepackage{txfonts}
\usepackage{graphicx}
\usepackage{amssymb}
\usepackage{float}
\usepackage{epstopdf}
\usepackage[a4paper, total={6in, 9in}]{geometry}
\usepackage{amsmath,bm}

\usepackage{latexsym,amsmath,amsfonts,amsthm,mathrsfs,amssymb,cite}
\usepackage[usenames]{color}
\usepackage{graphicx,epsfig}
\usepackage{amscd}
\usepackage{amsbsy}
\usepackage{bm}
\usepackage{color}
\usepackage{cite}
\usepackage[colorlinks=true,urlcolor=blue,
citecolor=red,linkcolor=blue,linktocpage,pdfpagelabels]{hyperref}
\newtheorem{thm}{Theorem}[section]

\newtheorem{lem}{Lemma}[section]

\theoremstyle{definition}
\newtheorem{definition}{Definition}[section]

\newtheorem{rem}{Remark}[section]
\numberwithin{equation}{section}

\numberwithin{equation}{section}

\newcounter{saveeqn}

\newcommand{\bdelta}{\bm{\delta}}

\newcommand{\Bx}{\mathbf{x}}
\newcommand{\By}{\mathbf{y}}

\newcommand{\Ocal}{\mathcal{O}}

\renewcommand{\(}{\left(}
\renewcommand{\)}{\right)}

\newcommand{\RR}{\mathbb{R}}

\title[Multi-layer elastic metamaterial structures and polariton resonances]
{Elastostatics with multi-layer metamaterial structures and an algebraic framework for polariton resonances}

\author{Youjun Deng}
\address{School of Mathematics and Statistics, Central South University, Changsha 410083, China}
\email{youjundeng@csu.edu.cn, dengyijun\_001@163.com}

\author{Lingzheng Kong}
\address{School of Mathematics and Statistics, Central South University, Changsha 410083, China}
\email{math\_klz@csu.edu.cn, math\_klz@163.com}

\author{Hongyu Liu}
\address{Department of Mathematics, City University of Hong Kong, Hong Kong SAR, China.}
\email{hongyu.liuip@gmail.com}

\author{Liyan Zhu}
\address{School of Mathematics and Statistics, Central South University, Changsha, Hunan, China.}
\email{mtah\_zly@163.com}

\date{} 
\begin{document}
\maketitle

\begin{abstract}

Multi-layer structures are ubiquitous in constructing metamaterial devices to realise various frontier applications including super-resolution imaging and invisibility cloaking. In this paper, we develop a general mathematical framework for studying elastostatics within multi-layer material structures in $\mathbb{R}^d$, $d=2,3$. The multi-layer structure is formed by concentric balls and each layer is filled by either a regular elastic material or an elastic metamaterial. The number of layers can be arbitrary and the material parameters in each layer may be different from one another. In practice, the multi-layer structure can serve as the building block for various material devices. Considering the impingement of an incident field on the multi-layer structure, we first derive the exact perturbed field in terms of an elastic momentum matrix, whose dimension is the same as the number of layers. By highly intricate and delicate analysis, we derive a comprehensive study of the spectral properties of the elastic momentum matrix. This enables us to establishe a handy algebraic framework for studying polariton resonances associated with multi-layer metamaterial structures, which forms the fundamental basis for many metamaterial applications.


\end{abstract}

\maketitle

\section{Introduction}

We are concerned with the mathematical study of elastic deformation within multi-layer solid structures and the peculiar resonance phenomena it induces. Specifically, we allow the presence of negative elastic materials, i.e. elastic materials with negative bulk moduli which form an important class of the so-called elastic metamaterials. In fact, exotic elastic materials with negative stiffness have been artificially structured \cite{KM14,LLBW01}, and have been proposed for many revolutionary applications including super-resolution elastic imaging and invisibility cloaking. In constructing metamaterial devices to realise those frontier applications, multi-layer structures are frequently used as the building blocks. Moreover, we would like to point out that the so-called polariton resonances induced by those negative structures form the fundamental basis of those applications.

In \cite{DLL201}, the authors studied the surface polariton resonance associated with a single elastic metamaterial particle in linear elasticity and considered its application in elastic wave imaging. This can be regarded as a single-layer metamaterial structure. In \cite{AKKY17,AKKY18,LL16,LLL16}, the core-shell structures were proposed for realising invisibility cloaking of elastic deformation based on the so-called anomalous localised resonance. Here, the core is filled with a regular elastic material whereas the shell is filled with a negative elastic material, and both the regular and negative elastic materials are uniform. If the material parameters in the core and the shell match in a delicate way involving the geometric configurations as well, then anomalous localised resonance can occur which in turn induces the cloaking effect due to an impinging field. The core-shell device can be regarded as a two-layer metamaterial structure. The aforementioned studies have also been extended from the micro-scale to the macro-scale in \cite{DLL19,DLL202,LLZ23}. Similarly, one-layer or two-layer metamaterial structures have also been extensively used and theoretically investigated in optics and acoustics; see e.g. \cite{ACKLM1,ADKLMZ,BS11,DLZ21,DLZ22,FDC19,FDL15,LLLW19,LL15,LLL9,RS,WN10} and the references cited therein; and in particular in \cite{DFArxiv}, multi-layer metamaterial structures have been studied in electrostatics. Motivated by those studies, we consider elastostatic deformation within much more general layered metamaterial structures in $\mathbb{R}^d$, $d=2,3$. In fact, we consider the case that the multi-layer structure is formed by concentric balls and each layer is filled by either a regular elastic material or an elastic metamaterial. The number of layers can be arbitrary and the material parameters in each layer may be different, though uniform. This can cover many of the existing studies mentioned above with the number of layers being 1 or 2. Then we consider the polariton resonance that can be induced by such general layered elastic metamaterial structures. To that end, we first derive the elastostatic field within the multi-layer structure due to an impinging field in terms of the so-called elastic momentum matrix, whose dimension is the same as the number of the layers. By highly intricate and delicate analysis, we derive a comprehensive study of the spectral properties of the elastic momentum matrix. It in turns establishes a convenient and handy algebraic framework for studying polariton resonances associated with multi-layer metamaterial structures. In practice, the multi-layer structure can serve as the building block for various material devices, and our algebraic framework can facilitate the proper selection of the metamaterial parameters in order to induce the desired resonance and in turn the realisation of customised applications. We shall investigate along this direction in our forthcoming work.

There is one more motivation and physical relevance of our study which can be described as follows. In fact, as mentioned earlier, the multi-layer structure in our study can be completely occupied by regular elastic materials, i.e. there is no presence of metamaterials. Such regular multi-layer structures have been proposed for achieving the so-called GPT-vanishing (Generalised Polarisation Tensors) or ESC-vanishing (Elastic Scattering Coefficients) structures and hence cloaking devices with enhanced invisibility effects via the transformation approach; see \cite{AAHWY17,LTWW21} for the related study in the elastic case, \cite{AKLLY13} in the electromagnetic case, \cite{AKLL11} in the electrostatic case and \cite{AKLL13} in the acoustic case. The key technical ingredient in those studies is an algebraic system defined by certain integral operators whose coefficient matrix is exactly the momentum matrix derive in this paper in the elastic case. The explicit formula of the characteristic polynomial and the estimation of the corresponding roots of the momentum matrix are set to be open problems in \cite[Page 258]{AKLL11} and \cite[Page 497]{AKLL13}:
\begin{enumerate}
	\item \emph{ However, as the number of layers gets larger, solving analytically the characteristic polynomial  seems too
		complicated, and even proving existence of solutions to the characteristic polynomial  seems to be quite challenging.
		These numerical evidences show us that the characteristic polynomial has solutions, even though we are not able to prove it.
	}
	\item \emph{As in the conductivity case \cite{AKLL11}, it should be emphasized that one does not know if a solution exists for any number of layers. }
\end{enumerate}
Due to such a reason, in the works mentioned above,  the authors only considered the numerical verification of the existence of such multi-layer structures with a small number of layers. Clearly, in the current article, we can completely address these open problems, at least in the elastostatic case. This paves the way for constructing much more general GPT-vanishing structures for enhanced invisibility cloaking devices by following a similar spirit as the works mentioned above. However, we choose to explore along that direction somewhere else and to focus on the development of the algebraic framework for the polariton resonance in the current paper.

The remainder of this paper is organized as follows.
In Section \ref{sec2}, we first present the elastostatic scattering problem with multi-layer structures.
Second, we give the representation of the perturbed field in terms of the elastic moment matrix.
In Section \ref{sec3}, we establish an algebraic framework for polariton resonances  and give the main results on the explicit formula of the characteristic polynomial and the estimation of the roots.
Section \ref{sec4} is devoted to the proofs of the main results in Section \ref{sec3}.
In Section \ref{sec5},
we give some similar results  in the two dimensional setting.
In Section \ref{sec6}, numerical examples are presented in finding all the polariton resonance modes in a fixed multi-layer structure and polariton resonance is simulated. Some conclusions are made in Section \ref{sec7}.

\section{Elastostatics with multi-layer metamaterial structures}\label{sec2}
\subsection{Multi-layer metamaterial structures}
We give the general $N$-layer elastic structure as shown in FIGURE \ref{fig:1}.
\begin{figure}[!h]
	\begin{center}
		{\includegraphics[width=2.5in]{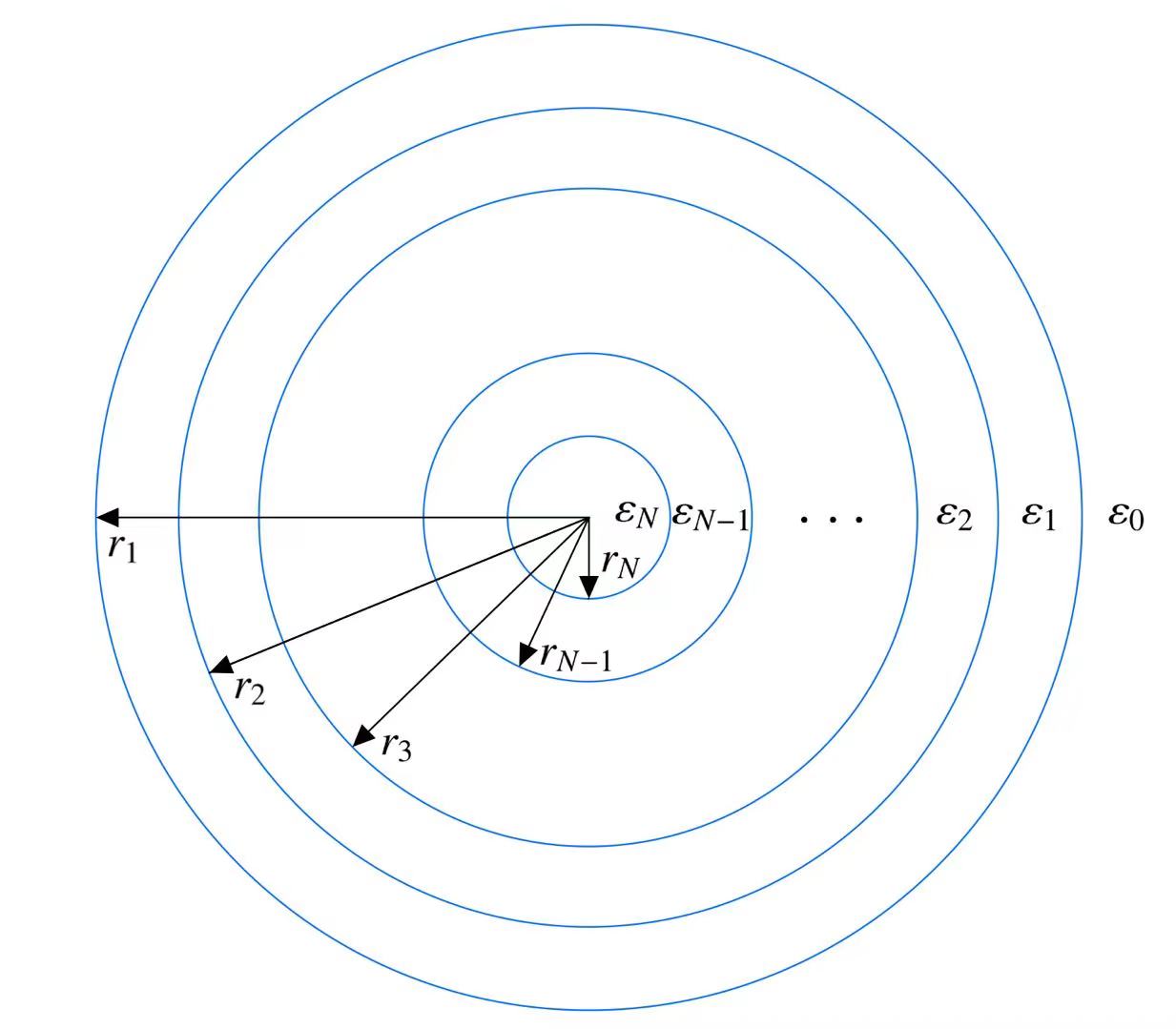}}
	\end{center}
	\caption{Schematic illustration of an $N$-layer elastic structure.
	}\label{fig:1}
\end{figure}
Precisely, we construct a sequence of layers, $D_0,D_1,\ldots,D_{N}$, by
\begin{equation}\label{eq:aj}
D_{0}:=\{r>r_{1}\}, \quad D_j:=\{r_{j+1}<r\leqslant r_{j}\}, \quad  j=1,2,\ldots, N-1, \quad D_N:=\{r\leqslant r_N\},
\end{equation}
and denote the interfaces between the adjacent layers by
\[
S_j:=\left\{|x|=r_j\right\}, \quad  j=1,2,\ldots,N,
\]
where $N\in \mathbb{N}$ and $r_j\in\mathbb{R}_+$. The elastic structure configuration in $\RR^d$, $d=2,3$, is characterized by the Lam\'e parameters $\tilde{\lambda}$ and $\tilde{\mu}$ which are given by
\begin{equation}\label{eq:def01}
(\tilde{\lambda},\tilde{\mu})= \varepsilon(\Bx)\cdot(\lambda,\mu)\quad\mbox{and}\quad \varepsilon(\Bx)=\varepsilon_c(\Bx)\chi(D)+\varepsilon_0\chi(\mathbb{R}^d\backslash \overline{D}),
\end{equation}
where $\chi$ denotes the characteristic function; $\mathbb{R}^d\backslash \overline{D}$ is the matrix, namely the background space; and $D=\cup_{j=1}^ND_j$
is the $N$-layer elastic structure. In \eqref{eq:def01}, $\lambda$ and $\mu$ signify the bulk moduli of a regular elastic material which are real constants satisfying the following strong convexity conditions,
\begin{equation}\label{eq:convexity condition}
	\mu>0 \quad \mbox{and} \quad d\lambda + 2\mu>0.
\end{equation}
$(\lambda, \mu)$ serves as the reference elastic configuration in our study. $\varepsilon_0$ is set to be a positive constant and hence $\varepsilon_0\cdot(\lambda, \mu)$ defines a regular elastic material in the matrix $\mathbb{R}^d\backslash\overline{D}$. For the $N$-layer structure, we set $\varepsilon_c$ to be of the following form
\begin{equation}\label{eq:paracho01}
\varepsilon_c(\Bx)=\varepsilon_j, \quad \Bx\in D_j,\quad j=1,2,\ldots,N.
\end{equation}
where
\begin{equation}\label{eq:pppn1}
\varepsilon_j\in\mathbb{C}\ \ \mbox{with}\ \ \Im\varepsilon_j\geq 0,\quad j=1,2,\ldots,N.
\end{equation}
That is, $\varepsilon_c$ is of a layered-piecewise constant structure. If $\Re\varepsilon_j>0$, then $\varepsilon_j\cdot(\lambda, \mu)$ defines a regular elastic material in $D_j$ with the bulk moduli being $\Re\varepsilon_j\cdot(\lambda, \mu)$ which clearly fulfils the strong convexity conditions in \eqref{eq:convexity condition}, and $\Im\varepsilon_j\cdot(\lambda, \mu)$ signifying the lossy/dampling parameters of the material. If $\Re\varepsilon_j\leq 0$, then $\varepsilon_j\cdot(\lambda, \mu)$ defines an elastic metamaterial in $D_j$ with the bulk moduli $\Re\varepsilon_j\cdot(\lambda, \mu)$ clearly broken the convexity conditions \eqref{eq:convexity condition}. In this paper, we simply refer to $\varepsilon_j\cdot(\lambda, \mu)$ with $\Re\varepsilon_j\leq 0$ as a negative elastic material in $D_j$.

In summary, we consider a rather general multi-layer structure where the number of layers can be arbitrarily given and the material in each layer can either be regular or negative. Moreover, the material parameters in each layer can be different from one another.

\subsection{Elastostatic scattering and momentum matrix}
Let $\mathbf{C}_{\lambda,\mu}(\Bx)  =(\mathrm{C}_{ijkl}(\Bx))_{i,j,k,l=1}^d$ be a four-rank tensor such that
\[
\mathbf{C}_{\tilde\lambda,\tilde\mu}=(\mathrm{C}_{ijkl}),\quad \mathrm{C}_{ijkl}(\Bx):=\tilde\lambda(\Bx)\bdelta_{ij}\bdelta_{kl}+\tilde\mu(\Bx)(\bdelta_{ik}\bdelta_{jl}+\bdelta_{il}\bdelta_{jk}),\ \ \Bx=(\Bx_1,\ldots,\Bx_d)\in\mathbb{R}^d,
\]
where $\bm{\delta}$ is the Kronecker delta, and $\tilde\lambda$ and $\tilde\mu$ are the elastic moduli defined in \eqref{eq:def01}.
The corresponding Lam\'e operator $\mathcal{L}_{\tilde{\lambda},\tilde{\mu}}$ is defined by
\begin{equation}\label{lameoperator}
	\mathcal{L}_{{\tilde\lambda},\tilde{\mu}} \mathbf{u}:=\nabla\cdot\mathbf{C}_{\tilde\lambda,\tilde\mu} {{\nabla}^s}\mathbf{u}=\tilde{\mu}\Delta\mathbf{u}+(\tilde{\lambda}+\tilde{\mu})\nabla\nabla\cdot\mathbf{u},
\end{equation}
with the strain tensor
\[
{\nabla}^s\mathbf{u}=\frac{1}{2}(\nabla \mathbf{u}+\nabla\mathbf{u}^T),
\]
where $\mathbf{u}$ is a $\mathbb{C}^d$-valued function, which signifies the elastic deformation field, and $T$ denotes the the transpose of a matrix.

Associated with the multi-layer elastic structure configuration described by \eqref{eq:def01}--\eqref{eq:paracho01}, the elastostatic scattering is governed by the following Lam\'e system:
\begin{equation}\label{eq:mainmd01}
\left\{
\begin{array}{ll}
\mathcal{L}_{\tilde{\lambda},\tilde{\mu}} \mathbf{u} =0, & \mbox{in} \quad \RR^d,\\
\mathbf{u}-\mathbf{H}=\Ocal(|\Bx|^{1-d}), & |\Bx|\rightarrow \infty,
\end{array}
\right.
\end{equation}
where the displacement field $\mathbf{u}\in H^1_{loc}(\RR^d)^d$, and the background field $\mathbf{H}$ satisfies $\mathcal{L}_{{\lambda},{\mu}} \mathbf{H}=0$ in $\RR^d$ with $\mathcal{L}_{{\lambda},{\mu}}$ defined in \eqref{lameoperator} with the subscripts replaced by ${\lambda}$ and ${\mu}$.

In what follows, we shall confine our study of the electrostatic system \eqref{eq:mainmd01} within the three dimensional setting, and shall present the corresponding two-dimensional extension in Section~\ref{sec5}.
Recall that, in \cite[Lemma 3.2]{DLL19}, three types of vectorial polynomials are introduced:
\[
\mathcal{T}^m_n(\Bx)=\nabla (r^{n} Y^m_{n}(\hat{\Bx}))\times \Bx, \quad n\geqslant 1, \quad  -n\leqslant m\leqslant n,
\]
\[
\mathcal{M}^m_n(\Bx)=\nabla(r^{n} Y^m_{n}(\hat{\Bx})), \quad n\geqslant 1,\quad  -n\leqslant m\leqslant n,
\]
and
\[
\mathcal{N}^m_n(\Bx)=\gamma^m_n r^{n-1} Y^m_{n-1}(\hat{\Bx})\Bx + (1-\frac{\gamma_n^m}{2n-1}-r^2) \nabla(r^{n-1} Y^m_{n-1}(\hat{\Bx})),
\]
where
\[
\gamma^m_n = \frac{2(n-1)\lambda + 2(3n-2)\mu}{(n+2)\lambda + (n+4)\mu}, \quad n\geqslant 1, \quad  -(n-1)\leqslant m\leqslant n-1,
\]
which are solutions to the Lam\'e $\mathcal{L}_{{\lambda},{\mu}} \mathbf{u} =0$ in $\mathbb{R}^d$.
By straight forward computations, one can also verify that
\[
{\mathcal{\widetilde{T}}}^m_n(\Bx)=\nabla (r^{-n-1} Y^m_{n}(\hat{\Bx}))\times \Bx, \quad n\geqslant 1, \quad  -n\leqslant m\leqslant n,
\]
\[
{\mathcal{\widetilde{M}}}^m_n(\Bx)=\tilde\gamma^m_n r^{-n-2} Y^m_{n-1}(\hat{\Bx})\Bx + (1+\frac{\tilde\gamma_n^m}{2n+1}-r^2) \nabla(r^{-n-2} Y^m_{n-1}(\hat{\Bx})),\quad n\geqslant 1, \quad  -(n-1)\leqslant m\leqslant n-1,
\]
where
\[
\tilde\gamma^m_n = \frac{2(n+1)\lambda + (6n+11)\mu}{(n-1)\lambda + (n-5)\mu},
\]
and
\[
{\mathcal{\widetilde{N}}}^m_n(\Bx)=\nabla(r^{-n-1} Y^m_{n}(\hat{\Bx})), \quad n\geqslant 1,\quad  -n\leqslant m\leqslant n,
\]
are solutions to the elastic system $\mathcal{L}_{{\lambda},{\mu}} \mathbf{u} =0$ in $\mathbb{R}^d$ excluding the origin.

Since $\{\mathcal{T}_n^m, \mathcal{M}_n^m, \mathcal{N}_n^m\}$ forms a complete system, we can suppose that the background field $\mathbf{H}$ is represented by
\begin{equation}\label{eq:defH0101}
\mathbf{H}= \sum_{n=1}^{\infty}\sum_{m=-n}^n (\mathfrak{t}_{n,m}^0  \mathcal{T}_{n}^{m}+\mathfrak{m}_{n,m}^0  \mathcal{M}_{n}^{m}+\mathfrak{n}_{n,m}^0  \mathcal{N}_{n}^{m}).
\end{equation}
Then the displacement field $\mathbf{u}$ can be written as the following form:
\begin{equation}\label{eq:defu0101}
\mathbf{u}=\left\{
\begin{aligned}
&\sum_{n=1}^{\infty}\sum_{m=-n}^n(\mathfrak{t}_{n,m}^N  \mathcal{T}_{n}^{m}+\mathfrak{m}_{n,m}^N  \mathcal{M}_{n}^{m}+\mathfrak{n}_{n,m}^N  \mathcal{N}_{n}^{m}), & \Bx&\in D_N,\\
&\sum_{n=1}^{\infty}\sum_{m=-n}^n(\mathfrak{t}_{n,m}^j  \mathcal{T}_{n}^{m}+\mathfrak{m}_{n,m}^j  \mathcal{M}_{n}^{m}+\mathfrak{n}_{n,m}^j  \mathcal{N}_{n}^{m})& &\\
&+ \sum_{n=1}^{\infty}\sum_{m=-n}^n({\mathfrak{\tilde t}}_{n,m}^j  {\mathcal{\widetilde T}}_{n}^{m}+{\mathfrak{\tilde m}}_{n,m}^j  {\mathcal{\widetilde M}}_{n}^{m}+{\mathfrak{\tilde n}}_{n,m}^j  {\mathcal{\widetilde N}}_{n}^{m}), & \Bx&\in D_j,\quad  j=N-1, N-2, \ldots , 1,\\
&\mathbf{H}+\sum_{n=1}^{\infty}\sum_{m=-n}^n({\mathfrak{\tilde t}}_{n,m}^0  {\mathcal{\widetilde T}}_{n}^{m}+{\mathfrak{\tilde m}}_{n,m}^0  {\mathcal{\widetilde M}}_{n}^{m}+{\mathfrak{\tilde n}}_{n,m}^0  {\mathcal{\widetilde N}}_{n}^{m}), & \Bx&\in D_{0}.
\end{aligned}
\right.
\end{equation}
By using transmission conditions on each layer and making use of the orthogonality of the vectorial spherical harmonic functions, one can separate each term above and form different transmission equations. In order to ease the exposition, we next mainly consider the case that the background field contains only one type of modes with $n=1$ and $m=0$; see Remark~\ref{rem2.102} for more related discussion about this point.

Suppose the background field is uniformly distributed in $\RR^3$, i.e. the field $\mathbf{H}$ can be represented as
\begin{equation}\label{eq:defH01}
\mathbf{H}=  a_{0}  \mathcal{N}_{1}^{0}(\Bx).
\end{equation}
Based on the symmetric properties of the multi-layer structure, we assume that the displacement field $\mathbf{u}$ takes the form
\[
\mathbf{u}=\left\{
\begin{aligned}
&a_{N}\mathcal{N}_{1}^{0}, & \Bx&\in D_N,\\
&a_{j}\mathcal{N}_{1}^{0}+r^{-3}b_{j}\mathcal{N}_{1}^{0}, & \Bx&\in D_j,\quad  j=N-1, N-2, \ldots , 1.\\
&a_{0}\mathcal{N}_{1}^{0}+r^{-3}b_{0}\mathcal{N}_{1}^{0}, & \Bx&\in D_{0},
\end{aligned}
\right.
\]
The solution $\mathbf{u}$ satisfies the transmission conditions
\begin{equation}\label{TC}
\left.\mathbf{u}\right|_{-}=\left.\mathbf{u}\right|_{+} \quad \text { and }\left.\quad \varepsilon_j \frac{\partial \mathbf{u}}{\partial \bm{\nu}}\right|_{-}=\left.\varepsilon_{j-1} \frac{\partial \mathbf{u}}{\partial \bm{\nu}}\right|_{+} \quad \text { on }\quad S_j,
\end{equation}
for $j=1,2,\ldots,N$, where
the conormtal derivative (or traction) on $S_j$ is defined by
\begin{equation}\label{traction}
	\frac{\partial \mathbf{u}}{\partial \bm{\nu}}=\lambda(\nabla \cdot \mathbf{u}) \bm{\nu}+\mu\left(\nabla \mathbf{u}+\nabla \mathbf{u}^{T}\right) \bm{\nu}
\end{equation}
where $\bm{\nu}$ is the exterior unit normal vector to $S_j$
and
\[
\left.v\right|_{\pm}(\Bx)=\lim _{h \rightarrow 0^{+}} v(\boldsymbol{\Bx} \pm h \bm{\nu}), \quad \Bx \in S_j,
\]
for an arbitrary function $v$.
In what follows, we define
\begin{equation}\label{eq:deflamb01}
\beta_{j}=\frac{\frac{4\mu}{3\lambda+2\mu}\varepsilon_{j-1}+\varepsilon_{j}}{\varepsilon_{j-1}-\varepsilon_{j}}, \quad j=1,2,\ldots,N.
\end{equation}
Now we give the representation of the perturbed field $\mathbf u-\mathbf{H}$ outside the multi-layer structure.
\begin{thm}\label{th:solmain01}
	Suppose $\mathbf u$ is the solution to the Lam\'e system \eqref{eq:aj}--\eqref{eq:mainmd01} in $\RR^3$, with the parameter $\varepsilon_c$ given by  \eqref{eq:paracho01}. Let $\mathbf H$ be given by \eqref{eq:defH01}. Define the \emph{elastic momentum matrix} $\mathbb{P}_N$ as follows.
	\begin{equation}\label{eq:matP01}
	\mathbb{P}_{N}:= \begin{bmatrix}
	\beta_1 & -1 & -1 & \cdots& -1 \\
	\frac{4\mu}{3\lambda+2\mu}(r_{2}/r_1)^3 & \beta_{2} & -1 &\cdots & -1 \\
	\vdots & \vdots & \vdots & \ddots & \vdots\\
	\frac{4\mu}{3\lambda+2\mu}(r_{N-1}/r_{1})^3 & \frac{4\mu}{3\lambda+2\mu}(r_{N-1}/r_{2})^3 & \frac{4\mu}{3\lambda+2\mu}(r_{N-1}/r_{3})^3 & \cdots & -1 \\
	\frac{4\mu}{3\lambda+2\mu}(r_{N}/r_1)^3 & \frac{4\mu}{3\lambda+2\mu}(r_{N}/r_2)^3 & \frac{4\mu}{3\lambda+2\mu}(r_{N}/r_{3})^3 &\cdots & \beta_{N}
	\end{bmatrix}
	\end{equation}
	If $\mathbb{P}_N$ is invertible, then the Lam\'e system \eqref{eq:mainmd01} is uniquely solvable with the solution given by the following formula:
	\begin{equation}\label{eq:purbmn01}
	\mathbf u-\mathbf{H}=r^{-3}\bm{e}^T \Upsilon_{N} (\mathbb{P}_{N}^T)^{-1}\bm{e}\mathbf{H},
	\end{equation}
	where $\bm{e}:=(1,1,\ldots,1)^T$, and
	\begin{equation}\label{eq:matU01}
	\Upsilon_{N}:= \begin{bmatrix}
	r_1^3 & 0 & 0 &\cdots& 0 \\
	0 & r_{2}^3 & 0 &\cdots& 0 \\
	0 & 0 & r_{3}^3 &\cdots& 0 \\
	\vdots & \vdots & \vdots & \ddots & \vdots\\
	0 & 0 & 0 &\cdots & r_{N}^3
	\end{bmatrix} .
	\end{equation}
\end{thm}
\begin{rem}\label{rem2.1}
It is clear that the invertibility of the momentum matrix $\mathbb{P}_N$ is equivalent to the well-posedness of the Lam\'e system \eqref{eq:mainmd01}. If $\Re\varepsilon>0$, namely $\mathbf{C}_{\tilde\lambda,\tilde\mu}$ is a regular elastic medium, then $\mathcal{L}_{\tilde\lambda, \tilde\mu}$ is an elliptic PDO (Partial Differential Operator), which guarantees the well-posedness of \eqref{eq:mainmd01} and hence the invertibility of $\mathbb{P}_N$. On the other hand, if for some layers the elastic materials are negative, the ellipticity of $\mathbb{P}_N$ might be broken. Nevertheless, if for those negative materials, say $\varepsilon_j$, if one has $\Im\varepsilon_j>0$, then the ellipticity of $\mathcal{L}_{\tilde\lambda,\tilde\mu}$ is retained and hence the invertibility of $\mathbb{P}_N$. This can be seen in our subsequent analysis.
\end{rem}
\begin{rem}\label{rem2.102}
We remark that we only consider the uniform background field case to derive the exact perturbation formula \eqref{eq:purbmn01}. However, even if the background field contains more modes, one should be able to derive the related perturbation formula by using the form \eqref{eq:defu0101}. For an illustration, if $\mathbf{H}= \mathfrak{t}_{n,m}^0  \mathcal{T}_{n}^{m}$ for some specific $m$ and $n$, then the displacement field $\mathbf{u}$ can be written by
\[
\mathbf{u}=\left\{
\begin{aligned}
&\mathfrak{t}_{n,m}^N  \mathcal{T}_{n}^{m}, & \Bx&\in D_N,\\
&\mathfrak{t}_{n,m}^j  \mathcal{T}_{n}^{m}+{\mathfrak{\tilde t}}_{n,m}^j  {\mathcal{\widetilde T}}_{n}^{m}, & \Bx&\in D_j,\quad  j=N-1, N-2, \ldots , 1,\\
&\mathfrak{t}_{n,m}^0  \mathcal{T}_{n}^{m}+{\mathfrak{\tilde t}}_{n,m}^0  {\mathcal{\widetilde T}}_{n}^{m}, & \Bx&\in D_{0}.
\end{aligned}
\right.
\]
By using a similar strategy as the proof of Theorem \ref{th:solmain01} in what follows, one can derive the exact formula for the coefficient ${\mathfrak{\tilde t}}_{n,m}^0$. If the background field contain other modes, say e.g. $\mathcal{M}_{n}^{m}$, the representation formula can be more complicated, but it will not bring any essential difficulty. For the sake of simplicity and to better express our main idea, we only study the case for a uniformly distributed background field.
\end{rem}
\subsection{Proof of Theorem \ref{th:solmain01}}
By using the transmission conditions \eqref{TC} across the interface $S_j$, $j=1, 2, \ldots N$, we can deduce the following equations
\begin{equation}\label{eq:trans01}
\left\{
\begin{aligned}
&a_{j}  + b_{j}r_j^{-3} = a_{j-1} + b_{j-1}r_{j}^{-3} , \\
&\varepsilon_{j} \left((3\lambda+2\mu) a_{j}  - 4\mu b_{j} r_j^{-3}
\right) = \varepsilon_{j-1} \left((3\lambda+2\mu) a_{j-1}  - 4\mu b_{j-1} r_{j}^{-3}
\right)
\end{aligned}
\right.
\end{equation}
where we set $b_{N}=0$.
By using \eqref{eq:deflamb01} and some basic arrangements to the equations \eqref{eq:trans01}, we obtain that
\begin{align*}
&\beta_1(a_{1}-a_{0})+\frac{4\mu}{3\lambda+2\mu}\sum_{j=2}^{N} (a_{j}-a_{j-1})\Big(\frac{r_j}{r_1}\Big)^3=a_{0}, \\
&-\sum_{j=1}^{l-1}(a_{j}-a_{j-1})+\beta_l(a_{l}-a_{l-1})+\frac{4\mu}{3\lambda+2\mu}\sum_{j=l+1}^{N}(a_{j}-a_{j-1})\Big(\frac{r_j}{r_l}\Big)^3=a_{0}, \quad l=2,3,\ldots,N-1,\\
&-\sum_{j=1}^{N-1}(a_{j}-a_{j-1})+\beta_{N}(a_{N}-a_{N-1})=a_{0}.
\end{align*}
The matrix $\mathbb{P}_N$ is invertible in the case that the elastic material parameters $\varepsilon_j$, $j=1,2,\ldots N$, are all positive,  and then there holds that
\begin{equation}\label{eq:a}
\bm{a} = a_{0}(\Xi(\mathbb{P}_N^T)^{-1}\bm{e} +\bm{e}),
\end{equation}
where $\bm{a}:=(a_{1},a_{2},\ldots,a_{N})^T$ and $\Xi$ is defined by
\begin{equation}\label{eq:defxi01}
\Xi=
\begin{bmatrix}
1 & 0 & 0 &\cdots & 0 \\
1 & 1 & 0 &\cdots & 0 \\
\vdots & \vdots &\vdots &\ddots & \vdots\\
1 & 1  & 1 &\cdots & 1
\end{bmatrix}.
\end{equation}
Furthermore, by the first equation in \eqref{eq:trans01} and $b_N=0$, we have
\begin{equation}\label{eq:link_a_b}
\sum_{j=k}^{N}(a_j-a_{j-1})r_j^3=b_{k-1},\quad k=1,2,\ldots,N.
\end{equation}
Combining \eqref{eq:a} with \eqref{eq:link_a_b}, there holds
\begin{equation}\label{eq:b}
\bm b = a_{0}\Xi^T \Upsilon_{N} (\mathbb{P}_N^T)^{-1} \bm{e},
\end{equation}
where $\bm b:=(b_{0},b_{1},  \ldots ,b_{N-1})^T$. By extracting the first element in $\bm b$, we derive that
\[
\mathbf u-\mathbf{H}=r^{-3}b_{0}\mathcal{N}_{1}^{0}=r^{-3}a_{0}\bm{e}^T \Upsilon_{N} (\mathbb{P}_N^T)^{-1} \bm{e}\mathcal{N}_{1}^{0}=r^{-3}\bm{e}^T \Upsilon_{N} (\mathbb{P}_N^T)^{-1}\bm{e}\mathbf{H}.
\]
The proof is complete.
\section{An algebraic framework for polariton resonances}\label{sec3}
In this section, we shall establish an algebraic framework for polariton resonances. We mention that in order to analyze the polariton resonance phenomena, one should consider the situation that the elastic material is not regular, that is \eqref{eq:convexity condition} does not hold any more.
\subsection{Polariton resonance}
It is known that polariton resonance is usually relevant to some eigenvalue problem arising from the PDE systems. We shall also consider the related eigenvalue problem associated with elastostatics system for multi-layer structure.
According to our earlier discussion in Remark \ref{rem2.1}, if the parameters $\varepsilon_j$, $j=1, 2, \ldots, N$,  are all positive real-valued,  the elastostatic system \eqref{eq:aj}--\eqref{eq:mainmd01} has only trivial solution if $\mathbf{H}=0$. We seek non-trivial solutions to \eqref{eq:aj}--\eqref{eq:mainmd01} when the parameters  are allowed to be negative valued, i.e., $\varepsilon_{j}<0$ for some $j=1, 2, \ldots, N$.
In order to simplify the analysis, in our subsequent study, we always assume that
\[
\xi :=\frac{4\mu}{3\lambda+2\mu}\quad \mbox{and}\quad t^i_{j}:=(r_{j}/r_{i})^3, \quad i, j=1, 2, \ldots, N,
\]
and
\begin{equation}\label{eq:vepdef01}
\varepsilon_j
=\left\{
\begin{array}{ll}
-\varepsilon^*+\mathrm{i}\delta, & j \quad \mbox{is odd},\\
\varepsilon_0, & j \quad \mbox{is even},
\end{array}
\right.
\end{equation}
where $\varepsilon^*$ is a positive number to be chosen and $\delta>0$  is sufficiently small, signifying a lossy parameter. $\mathrm{i}=\sqrt{-1}$. Define
\begin{equation}\label{eq:defbeta011}
\beta=\frac{\xi\varepsilon_{0}-\varepsilon^*+\mathrm{i}\delta}{\varepsilon_{0}+\varepsilon^*-\mathrm{i}\delta}.
\end{equation}
Then one can readily see that
\[
\beta_j=
\left\{
\begin{array}{ll}
\beta, & j \quad \mbox{is odd},\\
\xi-1-\beta, & j \quad \mbox{is even},
\end{array}
\right.
\]
for $j=1,2,\ldots, N.$
Hence, the perturbed field \eqref{eq:purbmn01} can be rewritten by
\begin{equation}\label{eq:eigenspan01}
\mathbf u-\mathbf H=r^{-3}\bm{e}^T \Upsilon_{N} (\beta I-\mathbb{K}_{N}^T)^{-1} \tilde{\bm{e}}\mathbf H,
\end{equation}
where $\tilde{\bm{e}}:=(1,-1, 1, \dots, (-1)^{N-1})^T$ and the matrix $\mathbb{K}_N$ is given by
\begin{equation}
\mathbb{K}_N= \begin{bmatrix}
0 & -1 & 1 & \cdots& (-1)^{N-2}& (-1)^{N-1} \\
-\xi t^1_{2} & \xi-1 & 1 &\cdots & (-1)^{N-2}& (-1)^{N-1} \\
-\xi t^1_{3} & \xi t^2_{3} & 0 &\cdots & (-1)^{N-2}& (-1)^{N-1}\\
\vdots & \vdots & \vdots & \ddots & \vdots& \vdots\\
-\xi t^1_{N-1} & \xi t^2_{N-1} & -\xi t^3_{N-1} & \cdots & (\xi-1+(-1)^{N-1}(\xi-1))/2& (-1)^{N-1} \\
-\xi t^1_{N} & \xi t^2_{N} & -\xi t^3_{N} &\cdots & (-1)^{N}\xi t^{N-1}_{N}& (\xi-1+(-1)^{N}(\xi-1))/2
\end{bmatrix}.
\end{equation}

Next, we  give the definition of the polariton resonance for our subsequent study.
\begin{definition}\label{df:def01}
	Consider the system \eqref{eq:mainmd01} associated with the $N$-layer structure $D$, where the elastic material configuration is described in \eqref{eq:vepdef01}. Polariton resonance occurs if there exists elastic material configuration such that, in the limiting case as the loss parameter $\delta$ goes to zero, the corresponding PDO $\mathcal{L}_{\tilde{\lambda},\tilde{\mu}}$  possesses non-trivial kernel which in turn induces resonance.
\end{definition}
\begin{thm}\label{th:resonancedefth01}
Suppose $\mathbf u$ is the solution to the Lam\'e system \eqref{eq:aj}--\eqref{eq:mainmd01} in $\RR^3$, with the parameter $\varepsilon_c$ be given in \eqref{eq:vepdef01}. If there holds
$$
\lim_{\delta\rightarrow 0} \beta = \beta^*,
$$
where $\beta^*$ is an eigenvalue of $\mathbb{K}_N$, then the polariton resonance occurs.
\end{thm}
\begin{proof}
According to Definition \ref{df:def01}, one only need to verify that the corresponding PDO $\mathcal{L}_{\tilde{\lambda},\tilde{\mu}}$  possesses a non-trivial kernel as $\delta \rightarrow 0$. To this end, let us consider the following Lam\'e system:
\begin{equation}\label{eq:lametrivi01}
\left\{
\begin{array}{ll}
\mathcal{L}_{\tilde{\lambda},\tilde{\mu}} \mathbf{u} =0, & \mbox{in} \quad \RR^d,\\
\mathbf{u}=\Ocal(|\Bx|^{1-d}), & |\Bx|\rightarrow \infty.
\end{array}
\right.
\end{equation}
We shall search for the solution of the following form:
\[
\mathbf{u}=\left\{
\begin{aligned}
&a_{N}\mathcal{N}_{1}^{0}, & \Bx&\in D_N,\\
&a_{j}\mathcal{N}_{1}^{0}+r^{-3}b_{j}\mathcal{N}_{1}^{0}, & \Bx&\in D_j,\quad  j=N-1, N-2, \ldots , 1,\\
&r^{-3}b_{0}\mathcal{N}_{1}^{0}, & \Bx&\in D_{0}.
\end{aligned}
\right.
\]
By following a similar treatment as the proof of Theorem \ref{th:solmain01}, one can show that
\begin{align*}
&\beta_1(a_{1})+\frac{4\mu}{3\lambda+2\mu}\sum_{j=2}^{N} (a_{j}-a_{j-1})\Big(\frac{r_j}{r_1}\Big)^3=0, \\
&-\sum_{j=1}^{l-1}(a_{j}-a_{j-1})+\beta_l(a_{l}-a_{l-1})+\frac{4\mu}{3\lambda+2\mu}\sum_{j=l+1}^{N}(a_{j}-a_{j-1})\Big(\frac{r_j}{r_l}\Big)^3=0, \quad l=2,3,\ldots,N-1,\\
&-\sum_{j=1}^{N-1}(a_{j}-a_{j-1})+\beta_N(a_{N}-a_{N-1})=0.
\end{align*}
That is
\begin{equation}\label{eq:notri01}
(\beta I-\mathbb{K}_{N}^T)\Sigma^{-1}\bm{a}=0,
\end{equation}
where $\Sigma$ and $\bm{a}$ are defined in proof of Theorem \ref{th:solmain01}. By letting $\delta\rightarrow 0$, one has $|\beta I-\mathbb{K}_N^T|=0$ and thus the homogeneous linear equations \eqref{eq:notri01} has non-trivial solutions. Note that $\Sigma^{-1}\bm{a}=(a_1, a_2-a_1, \ldots, a_N- a_{N-1})$, one immediately has that $a_j$, $j=1, \ldots, N$ are not identically zero and thus the system \eqref{eq:lametrivi01} has non-trivial solutions.
\end{proof}

In the following, we shall find out the suitable elastic material configurations to make the PDO $\mathcal{L}_{\tilde{\lambda},\tilde{\mu}}$ possess a non-trivial kernel by establishing an algebraic framework.

\subsection{Resonant fields and energy blowup}
For $\mathbf{u},\mathbf{v}\in H^1_{\text{loc}}(\mathbb{R}^3)^3$ , we introduce a bilinear form associated with the Lam\'e parameters $\lambda,\mu$ by
\begin{equation}\label{eq:energy1}
\begin{aligned}
\langle\mathbf{u},\mathbf{v}\rangle_{\mathbb{R}^3}^{\lambda,\mu}:&=\int_{\mathbb{R}^3}\nabla^{s} \mathbf{u}(\Bx): {\mathbf{C}_{\lambda,\mu}} \nabla^{s} \mathbf{v}(\Bx)~ \mathrm{d} \mathbf{x}\\
&=\int_{\mathbb{R}^3}\big[\lambda(\nabla\cdot\mathbf{u}){(\nabla\cdot\mathbf{v})}(\Bx)+2\mu\nabla^s\mathbf{u}(\Bx):{\nabla^s\mathbf{v}(\Bx)} \big]~ \mathrm{d} \Bx
\end{aligned}
\end{equation}
where and also in what follows, $\mathbf{A}:\mathbf{B}=\sum_{i,j=1}^3 a_{ij}b_{ij}$ for two matrices $\mathbf{A}=(a_{ij})_{i,j=1}^3$ and $\mathbf{B}=(b_{ij})_{i,j=1}^3$. The corresponding quadratic form is defined by
\begin{equation}\label{quadratic}
\mathbf{E}_{\mathbb{R}^3}^{\lambda, \mu}(\mathbf{u}):=\langle\mathbf{u}, \mathbf{u}\rangle_{\mathbb{R}^3}^{\lambda, \mu}.
\end{equation}

The resonant field  demonstrates an energy blowup.
Next, similar to \cite{DLZ21,LL16}, we give the formal definition of the polariton resonance.
\begin{definition}\label{df:def02}
	Consider the system \eqref{eq:mainmd01} associated with the $N$-layer structure $D$, where the elastic material configuration is described in \eqref{eq:vepdef01}.
	Then polariton resonance occurs if the following condition is fulfilled:
	\[
	\lim_{\delta\rightarrow +0}
	\mathbf{E}_{\RR^3\setminus\overline{D}}^{\lambda, \mu}(\mathbf{u}-\mathbf{H}):=\lim_{\delta\rightarrow +0}
	\langle\mathbf{u}-\mathbf{H},\mathbf{u}-\mathbf{H}\rangle_{\RR^3\setminus\overline{D}}^{\lambda,\mu}=+\infty.
	\]
	where $\mathbf{E}_{\RR^3\setminus\overline{D}}^{\lambda, \mu}(\mathbf{u}-\mathbf{H})$	is given by \eqref{eq:energy1}--\eqref{quadratic} with the integration domain replaced by $\RR^3\setminus\overline{D}$.
\end{definition}

In view of Definition \ref{df:def02}, the polariton resonance mode is parallel to the elastic material configurations which make the parameter $\beta$ the eigenvalue of $\mathbb{K}_N$ as $\delta\rightarrow 0$. For this reason, we shall focus on the explicit computation of eigenvalues to the matrix $\mathbb{K}_N$, or the exact formula to the roots of $|\mathbb{P}_N|$.

\subsection{An algebraic framework}
In this subsection, we shall derive a precise connection between the polariton resonance and the choice of parameters $\varepsilon_c$ in the multi-layer structure $D$.
By observing the explicit formula of $|\mathbb{P}_N(\beta)|$, $4\leqslant N\leqslant 6$ (see Subsection \ref{subsec4.1} below for details), we  give the delicate result about the explicit formula of $|\mathbb{P}_N(\beta)|$ for all $N\in \mathbb{N}$ in the following theorem. For the convenience of description, we denote by $\mathbf{i}_{m}$ the multi-index $(i_1,i_2,\ldots,i_{m})$ and $\tau_{\mathbf{i}_{m}}$ its sign given by
\[
\tau_{\mathbf{i}_{m}}:=(-1)^{\sum_{j=1}^{m} i_j}.
\]
Moreover, we denote by $C^{i,m}_{n}$ the set of all combinations of $m$ out $n$, $m\leqslant n$, say e.g., for one combination
\[
(i_1, i_2, \ldots, i_{m})\in C^{i,m}_{n}\quad \mbox{satisfying} \quad
(i+1)\leqslant i_1,i_2,\ldots, i_{m}\leqslant (n+i),
\]  we order them in ascending way, that is, $i_1<i_2< \cdots< i_{m}$.
\begin{thm}\label{th:ellresult01}
Define
\[
q(\beta):=\beta^2-(\xi-1)\beta.
\]
Then for $N\in \mathbb{N}$, it holds that
\begin{equation}\label{eq:maincong01}
|\mathbb{P}_N(\beta)|=
\left\{
\begin{aligned}
&(-1)^L\left(\sum_{k=0}^L \xi^k q(\beta)^{L-k}\left(\sum_{\mathbf{i}_{2k}\in C_{N}^{0,2k}}\tau_{\mathbf{i}_{2k}}\prod_{l=1}^k t^{i_{2l-1}}_{i_{2l}}\right)\right), &  N=2L,  \\
&\beta(-1)^L\left(\sum_{k=0}^L \xi^k q(\beta)^{L-k}\left(\sum_{\mathbf{i}_{2k}\in C_{N}^{0,2k}}\tau_{\mathbf{i}_{2k}}\prod_{l=1}^k t^{i_{2l-1}}_{i_{2l}}\right)\right), & N=2L+1.
\end{aligned}
\right.
\end{equation}
\end{thm}
\begin{rem}
Theorem \ref{th:ellresult01} shows  that if the number of layers, $N$, is odd then there exists a zero root $\beta =0$ to the equation  $|\mathbb{P}_N(\beta)|=0$.
The other roots are contained in a quadratic polynomial, whose constant terms are the roots of a polynomial equation of order $\lfloor N/2\rfloor$. Here we denote by $\lfloor t \rfloor$ the integer part of $t\in \RR$.
In other words, to solve the equation $|\mathbb{P}_N(\beta)|=0$, we first solve the equation $|\mathbb{P}_N(q)|=0$, which is a polynomial equation of order $\lfloor N/2\rfloor$, with respect to $q$.
We then solve the following quadratic equation
\[
\beta^2-(\xi-1)\beta-q=0
\]
to find the roots of $|\mathbb{P}_N(\beta)|=0$.
\end{rem}

We shall analyze all the roots to the polynomial equation $|\mathbb{P}_N(\beta)|=0$. Note that the explicit formula \eqref{eq:maincong01} of $|\mathbb{P}_N(\beta)|$  can also be written by the following compact form:
\begin{equation} \label{eq:maincong02}
|\mathbb{P}_N(\beta)|=
(-1)^{\lfloor N/2\rfloor}\beta^{N-2\lfloor N/2\rfloor}\left(\sum_{k=0}^{\lfloor N/2\rfloor} \xi^k q(\beta)^{{\lfloor N/2\rfloor}-k}\left(\sum_{\mathbf{i}_{2k}\in C_{N}^{0,2k}}\tau_{\mathbf{i}_{2k}}\prod_{l=1}^k t^{i_{2l-1}}_{i_{2l}}\right)\right).
\end{equation}
As already mentioned before, in order to ensure that the polariton resonance occurs, it is essential to find the roots of the polynomial
\begin{equation}\label{eq:deffq01}
f_N(q):=\sum_{k=0}^{\lfloor N/2\rfloor} \xi^k q^{{\lfloor N/2\rfloor}-k}\left(\sum_{\mathbf{i}_{2k}\in C_{N}^{0,2k}}\tau_{\mathbf{i}_{2k}}\prod_{l=1}^k t^{i_{2l-1}}_{i_{2l}}\right).
\end{equation}
To this end, we have the following elementary result on the roots of $f_N(q)$:
\begin{thm}\label{th:ellresult02}
Let $f_N(q)$ be defined in \eqref{eq:deffq01} and $q^*$ be the root to $f_N(q)$.
Then there exists $\lfloor N/2\rfloor$ real values of roots to $f_N(q)$. Moreover, we have
\begin{equation}\label{eq:thmainsp01}
q^*\in \left[-\frac{(\xi-1)^2}{4}, \xi\right]=\left[-\frac{(3\lambda-2\mu)^2}{4(3\lambda+2\mu)^2},\frac{4\mu}{3\lambda+2\mu}\right].
\end{equation}
\end{thm}
\begin{rem}\label{rem3.2}\
	\indent
	\begin{enumerate}
		\item[(i)] It follows from Theorem \ref{th:ellresult02} that the roots of $|\mathbb{P}_N(\beta)|$ are all real values. In fact, for any real solution $q^*$ to \eqref{eq:deffq01}, one can derive that
		the polynomial equation $|\mathbb{P}_N(\beta)|=0$ has two roots,
		\[
		\beta=\frac{\xi-1\pm\sqrt{(\xi-1)^2+4q^*}}{2},
		\]
		this, togrther with \eqref{eq:thmainsp01},  implies that
		\[
		\beta \in [-1, \xi]=\left[-1,\frac{4\mu}{{3\lambda+2\mu}}\right].
		\]
		\item[(ii)] We also remark that if the polariton resonance occurs, the parameters $\varepsilon_{2j-1}$, $j=1,2,\ldots,N$,
		defined in \eqref{eq:vepdef01} should have a negative real part, i.e., $\varepsilon^*>0$.
		Indeed, if $\varepsilon^*<0$, by the definition of $\beta$ in \eqref{eq:defbeta011} we have that
		\[
		\Re\beta\in (-\infty,-1)\cup(\xi,+\infty).
		\]
		Thus in such a case, we cannot have the polariton resonance since the spectrum of the matrix $\mathbb{K}_N$ lies in $[-1,\xi]$.
	\end{enumerate}

\end{rem}

Next, we shall consider that the radius of the layers are extreme large.
\begin{thm}\label{co:ellresult02}
Suppose $r_i=R+c_i$, where $R\gg 1$ and $c_i$ are regular constants, $i=1, 2, \ldots, N$. Then the polynomial  \eqref{eq:deffq01} can be rewritten as
\begin{equation}
f_N(q)=
\sum_{k=0}^{\lfloor N/2\rfloor} (-1)^k \xi^k q^{{\lfloor N/2\rfloor}-k}C_{\lfloor N/2 \rfloor}^k+\Ocal(1/R)=(q-\xi)^{\lfloor N/2 \rfloor}+\Ocal(1/R).
\end{equation}
\end{thm}
\begin{rem}
This indicates that the possible polariton modes for such a set of multi-layer structure are only
$\beta=-1, \xi$ if the number of layers $N$ is even and $\beta=0, -1, \xi$ if $N$ is odd.
\end{rem}
\section{Proofs of Theorems \ref{th:ellresult01}--\ref{co:ellresult02}}\label{sec4}\
\newline
\indent
In this section, we study the polariton resonances for the elastostatic system, which in turn casts some light on deriving the explicit formula in Theorem \ref{th:ellresult01}, and hence facilitates the proofs of Theorems \ref{th:ellresult01}--\ref{co:ellresult02}.
\subsection{Eigenvalue problem}\label{subsec4.1}
As mentioned above, the polariton mode is parallel to the elastic material configurations which make the parameter $\beta$ the eigenvalue of $\mathbb{K}_N$ as $\delta\rightarrow 0$. So, we focus on the eigenvalues of the matrix $\mathbb{K}_N$, or the explicit formula  of $|\mathbb{P}_N|$ defined in \eqref{eq:matP01} in this subsection.
For this, we first define
\begin{equation}\label{eq:matP02}
P^i_{M}:= \begin{bmatrix}
\beta_i & -1 & -1 & \cdots& -1 \\
\frac{4\mu}{3\lambda+2\mu}(r_{i+1}/r_i)^3 & \beta_{i+1} & -1 &\cdots & -1 \\
\vdots & \vdots & \vdots & \ddots & \vdots\\
\frac{4\mu}{3\lambda+2\mu}(r_{M-1}/r_i)^3 & \frac{4\mu}{3\lambda+2\mu}(r_{M-1}/r_{i+1})^3 & \frac{4\mu}{3\lambda+2\mu}(r_{M-1}/r_{i+2})^3 & \cdots & -1 \\
\frac{4\mu}{3\lambda+2\mu}(r_{M}/r_i)^3 & \frac{4\mu}{3\lambda+2\mu}(r_{M}/r_{i+1})^3 & \frac{4\mu}{3\lambda+2\mu}(r_{M}/r_{i+2})^3 &\cdots & \beta_{M}
\end{bmatrix}
\end{equation}
and set $P^3_2=1$.

Next, we give the recursion formula for $|\mathbb{P}_N|$ in the following lemma.
\begin{lem}\label{le:main01}
	Let $N\geqslant 4$. Then there holds the following  recursion formula:
	\begin{equation}\label{eq:deter0102}
	\begin{aligned}
	|\mathbb{P}_N|&=\(\beta_1+(\beta_2-\xi+1)t^1_{2}\)\(\beta_{N}+(\beta_{N-1}-\xi+1)t^{N-1}_{N}\)|P^2_{N-1}| \\
	&\quad -\(\beta_1+(\beta_2-\xi+1)t^1_{2}\)(\beta_{N-1}+1)(\beta_{N-1}-\xi)t^{N-1}_{N}|P^2_{N-2}|\\
	&\quad-(\beta_{2}-\xi)t^1_{2}(\beta_{2}+1)\(\beta_{N}+(\beta_{N-1}-\xi+1)t^{N-1}_{N}\)|P^3_{N-1}|\\
	&\quad+(\beta_{2}-\xi)t^1_{2}(\beta_{2}+1)(\beta_{N-1}+1)(\beta_{N-1}-\xi)t^{N-1}_{N}|P^3_{N-2}|.
	\end{aligned}
	\end{equation}
\end{lem}
\begin{proof}
Denote by $\mathbf{E}_{i,j*k}$  the elementary matrix which is transform of identity matrix by multiplying each element of the $j$-th row of the identity matrix by $k$, and then adding it to the $i$-th row.
Thus, by using some elementary translation, we can obtain
	\begin{equation}
	\begin{aligned}
	|\mathbb{P}_N|= 
	&\left|\mathbf{E}_{N,(N-1)*(-t^{N-1}_N)}\mathbf{E}_{1,2*(-1)}\mathbb{P}_N\mathbf{E}^T_{1,2*(-t^1_2)}\mathbf{E}^T_{N,(N-1)*(-1)}\right|\\
	=&\left|\begin{array}{ccccc}
	\beta_1+(\beta_2-\xi+1)t^1_{2} & -\beta_{2}-1 & \cdots& 0& 0 \\
	(\xi-\beta_{2})t^1_{2} & \beta_{2} &\cdots& -1 & 0 \\
	\vdots & \vdots & \ddots &\vdots & \vdots\\
	0 & \xi t^2_{N-1} & \cdots& \beta_{N-1} & -\beta_{N-1}-1 \\
	0 & 0 &\cdots& (\xi-\beta_{N-1})t^{N-1}_{N} & \beta_{N}+(\beta_{N-1}-\xi+1)t^{N-1}_{N}
	\end{array}
	\right|,
	\end{aligned}
	\end{equation}
	where we only changed the first row and column, together with the last row and column. By virtue of the notation \eqref{eq:matP02} and the Laplace expansion theorem for determinant one thus has that
	\begin{equation}
	\begin{aligned}
	|\mathbb{P}_N|=&\left(\beta_1+(\beta_2-\xi+1)t^1_{2}\right)\left(\beta_{N}+(\beta_{N-1}-\xi+1)t^{N-1}_{N}\right)|P^2_{N-1}| \\
	&-\left(\beta_1+(\beta_2-\xi+1)t^1_{2}\right)(\beta_{N-1}+1)(\beta_{N-1}-\xi)t^{N-1}_{N}|P^2_{N-2}|\\
	&-(\beta_{2}-\xi)t^1_{2}(\beta_{2}+1)\left(\beta_{N}+(\beta_{N-1}-\xi+1)t^{N-1}_{N}\right)|P^3_{N-1}|\\
	&+(\beta_{2}-\xi)t^1_{2}(\beta_{2}+1)(\beta_{N-1}+1)(\beta_{N-1}-\xi)t^{N-1}_{N}|P^3_{N-2}|
	\end{aligned}
	\end{equation}
	holds for all $N\geqslant 4$, $N\in \mathbb{N}$.
\end{proof}

Next, we shall study the dependence of the determinant of  $P^i_m$ defined in \eqref{eq:matP02}, $i\leqslant m$, on $\beta$.
By direct computations, we first have that for $m\geqslant 1$,
\begin{equation}\label{eq:rec0101}
\begin{aligned}
|P^m_{m+1}(\beta)|&=-\left(\beta^2 -(\xi-1)\beta\right)+\xi t^{m}_{m+1}, \\
|P^{2m-1}_{2m+1}(\beta)|&=\left(-\left(\beta^2 -(\xi-1)\beta\right)+\left(\xi t^{2m-1}_{2m}+\xi t^{2m}_{2m+1}-\xi t^{2m-1}_{2m+1}\right)\right)\beta,\\
|P^{2m}_{2m+2}(\beta)|&=(\xi-1-\beta)^2\beta+(\xi-1-\beta)\left(\xi t^{2m}_{2m+1}+\xi t^{2m+1}_{2m+2}-\xi t^{2m}_{2m+2}\right)\\
&= \left(-\left(\beta^2-(\xi-1)\beta\right)+\left(\xi t^{2m}_{2m+1}+\xi t^{2m+1}_{2m+2}-\xi t^{2m}_{2m+2}\right)\right)(\xi-1-\beta)
.
\end{aligned}
\end{equation}
Combining the recursion formula \eqref{eq:deter0102} with \eqref{eq:rec0101} we obtain that
\begin{equation}\label{p4}
\begin{aligned}
|\mathbb{P}_{4}(\beta)|&=\beta\(1-t^1_{2}\)(\xi-1-\beta)\(1-t^{3}_{4}\)\(\beta( \xi-1-\beta)+\xi t^{2}_{3}\)\\
&\quad  -\beta\(1-t^1_{2}\)(\beta+1)(\beta-\xi)t^{3}_{4}(\xi-1-\beta)-(\beta+1)t^1_{2}(\beta-\xi)(\xi-1-\beta)\(1-t^{3}_{4}\)\beta\\
&\quad +(\beta+1)t^1_{2}(\beta-\xi)(\beta+1)(\beta-\xi)t^{3}_{4}\\
&=(\beta^2-(\xi-1)\beta)^2
+\xi(\beta^2-(\xi-1)\beta)\sum_{(i,j)\in C_4^{0,2}}(-1)^{i+j}t^i_j+\xi^2 t^1_2t^3_4.
\end{aligned}
\end{equation}
In a similar manner, by straightforward (but lengthy and tedious) we have that
\begin{equation}
\begin{aligned}
|\mathbb{P}_{5}(\beta)|
&=
\beta\(1-t^1_{2}\)\beta\(1- t^{4}_{5}\)\left(-\(\beta^2-(\xi-1)\beta\)+\(\xi t^{2}_{3}+\xi t^{3}_{4}-\xi t^{2}_{4}\)\right)(\xi-1-\beta) \\
&\quad -\beta\(1-t^1_{2}\)(\beta-\xi)(\beta+1)t^{4}_{5}\(\beta( \xi-1-\beta)+\xi t^{2}_{3}\)\\
&\quad-(\beta+1)t^1_{2}(\beta-\xi)\beta\(1- t^{4}_{5}\)\(\beta( \xi-1-\beta)+\xi t^{3}_{4}\)\\
&\quad+(\beta+1)t^1_{2}(\beta-\xi)(\beta-\xi)(\beta+1)t^{4}_{5}\beta\\
&=\beta\left(\left(\beta^2-(\xi-1)\beta\right)^2+\xi\left(\beta^2- (\xi-1)\beta\right)\sum_{(i,j)\in C_5^{0,2}}(-1)^{i+j}t^i_j+\xi^{2}\sum_{(i,j,k,l)\in C_5^{0,4}}\tau_{(i,j,k,l)}t^i_jt^k_l\right).
\end{aligned}
\end{equation}
Similarly,
\begin{equation}\label{p6}
\begin{aligned}
|\mathbb{P}_{6}(\beta)|&=-\left(\beta^2-(\xi-1)\beta\right)^3-\xi\left(\beta^2-(\xi-1)\beta\right)^2\sum_{(i,j)\in C_6^{0,2}}(-1)^{i+j}t^i_j\\
&\quad -\xi^2\left(\beta^2-(\xi-1)\beta\right)\sum_{(i,j,k,l)\in C_6^{0,4}}\tau_{(i,j,k,l)}t^i_jt^k_l+\xi^3 t^1_2t^3_4t^5_6.
\end{aligned}
\end{equation}
Observing the above results \eqref{p4}--\eqref{p6} gives us a clue to derive the exact formula in Theorem \ref{th:ellresult01}.  Next, we prove it by mathematical induction.
\begin{proof}[Proof of Theorem \ref{th:ellresult01}]
We shall use induction to prove Theorem \ref{th:ellresult01}.
In view of \eqref{eq:maincong01} we can apply \eqref{eq:rec0101} to deduce that  \eqref{eq:maincong01} holds for $1\leqslant N\leqslant 3$.
Suppose that \eqref{eq:maincong01} holds for all $N\leqslant N_0$, $N_0\geqslant 4$, we show that it also holds for $N=N_0+1$. Note that $q(\xi-1-\beta)=q(\beta)$, in what follows, we shall denote $q(\beta)$ by $q$  for simplicity.
\begin{enumerate}
	\item [\textbf{Case i}] $N_0$ is even. Since \eqref{eq:maincong01}  and \eqref{eq:deter0102} hold for all $N\leqslant N_0$ and the  fact that $N=N_0+1$ is odd,  we obtain
	\begin{equation}
	\begin{aligned}
	|\mathbb{P}_N(\beta)|&=\beta^2\left(1-t^1_{2}\right)\left(1- t^{N-1}_{N}\right)|P^2_{N-1}(\xi-1-\beta)| +\beta\left(1-t^1_{2}\right)(\xi-\beta)(\beta+1)t^{N-1}_{N}|P^2_{N-2}(\xi-1-\beta)|\\
	&\quad+(\beta+1)t^1_{2}(\xi-\beta)\beta\left(1-t^{N-1}_{N}\right)|P^3_{N-1}(\beta)|+(\beta+1)^2t^1_{2}(\xi-\beta)^2t^{N-1}_{N}|P^3_{N-2}(\beta)|\\
	&=\sum_{j=1}^4 d_j \mathcal{D}_j,
	\end{aligned}
	\end{equation}
	where we use the notations
	\[
	\begin{aligned}
	d_1&=\beta^2\left(1-t^1_{2}\right)\left(1-t^{N-1}_{N}\right)(-1)^{N_0/2-1}(\xi-1-\beta),& d_2&=\beta\left(1-t^1_{2}\right)(\xi-\beta)(\beta+1)t^{N-1}_{N}(-1)^{N_0/2-1}, \\
	d_3&=(\beta+1)t^1_{2}(\xi-\beta)\beta\left(1-t^{N-1}_{N}\right)(-1)^{N_0/2-1}, & d_4&=(\xi-\beta)^2(\beta+1)^2t^1_{2}t^{N-1}_{N}(-1)^{N_0/2}\beta,
	\end{aligned}
	\]
	and
	\begin{equation}\label{DJ}
	\mathcal{D}_j=\left(\sum_{k=0}^{N_0/2-1-\lfloor j/4\rfloor} \xi^k q^{N_0/2-k-1-\lfloor j/4\rfloor}\left(\sum_{\mathbf{i}_{2k}\in C_{N-2-\lfloor j/2\rfloor}^{\lfloor (j+1)/2\rfloor,2k}}\tau_{\mathbf{i}_{2k}}\prod_{l=1}^k t^{i_{2l-1}}_{i_{2l}}\right)\right).
	\end{equation}
	From \eqref{DJ}, we set
	\[
	\begin{aligned}
	&\mathfrak{A}_k:=\sum_{\mathbf{i}_{2k}\in C_{N-2}^{1,2k}}\tau_{\mathbf{i}_{2k}}\prod_{l=1}^k t^{i_{2l-1}}_{i_{2l}},
	\quad \mathfrak{B}_k:=\sum_{\mathbf{i}_{2k}\in C_{N-3}^{1,2k}}\tau_{\mathbf{i}_{2k}}\prod_{l=1}^k t^{i_{2l-1}}_{i_{2l}},\\
	&\mathfrak{C}_k:=\sum_{\mathbf{i}_{2k}\in C_{N-3}^{2,2k}}\tau_{\mathbf{i}_{2k}}\prod_{l=1}^k t^{i_{2l-1}}_{i_{2l}},
	\quad \mathfrak{D}_k:=\sum_{\mathbf{i}_{2k}\in C_{N-4}^{2,2k}}\tau_{\mathbf{i}_{2k}}\prod_{l=1}^k t^{i_{2l-1}}_{i_{2l}}.
	\end{aligned}
	\]
	By direct computations, we deduce that
	\[
	\begin{aligned}
	|\mathbb{P}_N(\beta)|&=(-1)^{N_0/2}\beta\left(q\left(1-t^1_{2}\right)\left(1-t^{N-1}_{N}\right)\left(\sum_{k=0}^{N_0/2-1} \xi^k q^{N_0/2-k-1}\mathfrak{A}_k\right)\right.\\
	&
	\left.\quad\quad\quad\quad\quad+(q-\xi)\left(1-t^1_{2}\right)t^{N-1}_{N}\left(\sum_{k=0}^{N_0/2-1} \xi^k q^{N_0/2-k-1}\mathfrak{B}_k\right)\right.\\
	&
	\left.\quad\quad\quad\quad\quad+(q-\xi)t^1_{2}\left(1-t^{N-1}_{N}\right)\left(\sum_{k=0}^{N_0/2-1} \xi^k q^{N_0/2-k-1}\mathfrak{C}_k\right)\right.\\
	&
	\left.\quad\quad\quad\quad\quad+(q-\xi)^2t^1_{2}t^{N-1}_{N}\left(\sum_{k=0}^{N_0/2-2} \xi^k q^{N_0/2-k-2}\mathfrak{D}_k\right)\right)\\
	&=:(-1)^{N_0/2}\beta\left(\sum_{k=0}^{N_0/2}\xi^kq^{N_0/2-k}\mathrm{G}_k\right),
	\end{aligned}
	\]
	where it can be seen that
	\begin{equation}\label{G0}
	\mathrm{G}_0=\left(1+t^1_{2}t^{N-1}_{N}-t^1_{2}-t^{N-1}_{N}\right)+\left(1-t^1_{2}\right)t^{N-1}_{N}+t^1_{2}\left(1-t^{N-1}_{N}\right)+t^1_{2}t^{N-1}_{N}=1,
	\end{equation}
	and
	\begin{equation}\label{G1}
	\begin{aligned}
	\mathrm{G}_1&=\left(1+t^1_{2}t^{N-1}_{N}-t^1_{2}-t^{N-1}_{N}\right)\mathfrak{A}_1+\left(1-t^1_{2}\right)t^{N-1}_{N}\mathfrak{B}_1 +t^1_{2}\(1-t^{N-1}_{N}\)\mathfrak{C}_1+t^1_{2}t^{N-1}_{N}\mathfrak{D}_1\\
	&\quad - \(1-t^1_{2}\)t^{N-1}_{N}- t^1_{2}\(1-t^{N-1}_{N}\)-2 t^1_{2}t^{N-1}_{N}\\
	& =\sum_{\mathbf{i}_2\in C_{N}^{0,2}}\tau_{\mathbf{i}_2}t^{i_{1}}_{i_{2}},
	\end{aligned}
	\end{equation}
	where the last equality follows by using the relations
	\[
	\begin{aligned}
	\mathfrak{A}_1&=\mathfrak{B}_1+\sum_{j=2}^{N-2}(-1)^{j+N-1}t^j_{N-1},  &\mathfrak{A}_1&=\mathfrak{C}_1+\sum_{j=3}^{N-1}(-1)^{2+j}t^2_j,\\
	\mathfrak{B}_1&=\mathfrak{D}_1+\sum_{j=3}^{N-2}(-1)^{2+j}t^2_j,  &\mathfrak{C}_1&=\mathfrak{D}_1+\sum_{j=3}^{N-2}(-1)^{j+N-1}t^j_{N-1}.
	\end{aligned}
	\]
	We can also deduce that
	\begin{equation}\label{GN2}
	\begin{aligned}
	\mathrm{G}_{N_0/2}&=-\(1-t^1_{2}\)t^{N-1}_{N}\mathfrak{B}_{N_0/2-1}-t^1_{2}\(1-t^{N-1}_{N}\)\mathfrak{C}_{N_0/2-1}+t^1_2t^{N-1}_N\mathfrak{D}_{N_0/2-2}\\
	&=-\(1-t^1_{2}\)t^{N-1}_{N}\left(\sum_{\mathbf{i}_{N_0-2}\in C_{N-3}^{1,N_0-2}}\tau_{\mathbf{i}_{N_0-2}}\prod_{l=1}^{N_0/2-1} t^{i_{2l-1}}_{i_{2l}}\right)\\
	&\quad -t^1_{2}\(1-t^{N-1}_{N}\)\left(\sum_{\mathbf{i}_{N_0-2}\in C_{N-3}^{2,N_0-2}}\tau_{\mathbf{i}_{N_0-2}}\prod_{l=1}^{N_0/2-1} t^{i_{2l-1}}_{i_{2l}}\right)\\
	&\quad+t^1_2t^{N-1}_N\left(\sum_{\mathbf{i}_{N_0-4}\in C_{N-4}^{2,N_0-4}}\tau_{\mathbf{i}_{N_0-4}}\prod_{l=1}^{N_0/2-2} t^{i_{2l-1}}_{i_{2l}}\right)\\
	&=\(t^1_{2}t^{N-1}_{N}-t^{N-1}_{N}\)(-1)^{N_0/2-1}t^2_3t^4_5\cdots t^{N-3}_{N-2}+\(t^1_{2}t^{N-1}_{N}-t^{1}_{2}\)(-1)^{N_0/2-1}t^3_4t^5_6\cdots t^{N-2}_{N-1}\\
	&\quad+t^1_2t^{N-1}_N\left(\sum_{\mathbf{i}_{N_0-4}\in C_{N-4}^{2,N_0-4}}\tau_{\mathbf{i}_{N_0-4}}\prod_{l=1}^{N_0/2-2} t^{i_{2l-1}}_{i_{2l}}\right)\\
	&=\sum_{\mathbf{i}_{N_0}\in C_{N}^{0,N_0}}\tau_{\mathbf{i}_{N_0}}\prod_{l=1}^{N_0/2} t^{i_{2l-1}}_{i_{2l}}.
	\end{aligned}
	\end{equation}
	Furthermore, one can readily obtain that
	\begin{equation}\label{gk}
	\begin{aligned}
	\mathrm{G}_k&=\(1+t^1_{2}t^{N-1}_{N}-t^1_{2}-t^{N-1}_{N}\)\mathfrak{A}_k+\(1-t^1_{2}\)t^{N-1}_{N}(\mathfrak{B}_k-\mathfrak{B}_{k-1})\\
	&\quad +t^1_{2}\(1-t^{N-1}_{N}\)(\mathfrak{C}_k-\mathfrak{C}_{k-1})+t^1_{2}t^{N-1}_{N}(\mathfrak{D}_k-2\mathfrak{D}_{k-1}+\mathfrak{D}_{k-2}),
	\end{aligned}
	\end{equation}
	for all $k=2, 3, \ldots, N_0/2-1$.
	By some proper arrangements to the equation \eqref{gk}, we have
	\begin{equation}\label{eq:solpfmn01}
	\begin{aligned}
	\mathrm{G}_k&=\mathfrak{D}_k+\(1+t^1_{2}t^{N-1}_{N}-t^1_{2}-t^{N-1}_{N}\)\left(\mathfrak{A}_k-\mathfrak{C}_k-\mathfrak{B}_k+\mathfrak{D}_k\right)\\
	&\quad+\(1-t^1_{2}\)\left(\mathfrak{B}_k-\mathfrak{D}_k\right)+\(1-t^{N-1}_{N}\)\left(\mathfrak{C}_k-\mathfrak{D}_k\right)\\
	&\quad+\(t^1_2t^{N-1}_N-t^{N-1}_N\)\left(\mathfrak{B}_{k-1}-\mathfrak{D}_{k-1}\right)+\(t^1_2t^{N-1}_N-t^{1}_2\)\left(\mathfrak{C}_{k-1}-\mathfrak{D}_{k-1}\right)\\
	&\quad-\(t^1_2+t^{N-1}_N\)\mathfrak{D}_{k-1}+t^1_{2}t^{N-1}_{N}\mathfrak{D}_{k-2},
	\end{aligned}
	\end{equation}
	for $k=2, 3, \ldots, N_0/2-1$.
	By direct computations, there holds the following relations:
	\begin{equation}\label{eq:solpfmn02}
	\mathfrak{B}_k-\mathfrak{D}_k=\sum_{\mathbf{i}_{2k-1}\in C_{N-4}^{2,2k-1}}(-1)^2\tau_{\mathbf{i}_{2k-1}}t^2_{i_{1}}\prod_{l=1}^{k-1} t^{i_{2l}}_{i_{2l+1}},
	\quad \mathfrak{C}_k-\mathfrak{D}_k=\sum_{\mathbf{i}_{2k-1}\in C_{N-4}^{2,2k-1}}(-1)^{N-1}\tau_{\mathbf{i}_{2k-1}}\prod_{l=1}^{k-1} t^{i_{2l-1}}_{i_{2l}}t^{i_{2k-1}}_{N-1},
	\end{equation}
	and
	\begin{equation}\label{eq:solpfmn03}
	\mathfrak{A}_k-\mathfrak{C}_k-\mathfrak{B}_k+\mathfrak{D}_k=\sum_{\mathbf{i}_{2k-2}\in C_{N-4}^{2,2k-2}}(-1)^{N+1}\tau_{\mathbf{i}_{2k-2}}t^2_{i_{1}}\prod_{l=1}^{k-2} t^{i_{2l}}_{i_{2l+1}}t^{i_{2k-2}}_{N-1}.
	\end{equation}
	By substituting \eqref{eq:solpfmn02}--\eqref{eq:solpfmn03} into \eqref{eq:solpfmn01} and  straight forward computations,  we can deduce that
	\[
	\begin{aligned}
	\mathrm{G}_k&=\mathfrak{D}_k+\sum_{\mathbf{i}_{2k-2}\in C_{N-4}^{2,2k-2}}(-1)^{N+1}\tau_{\mathbf{i}_{2k-2}}t^2_{i_{1}}\prod_{l=1}^{k-2} t^{i_{2l}}_{i_{2l+1}}t^{i_{2k-2}}_{N-1}+\sum_{\mathbf{i}_{2k-2}\in C_{N-4}^{2,2k-2}}(-1)^{N+1}\tau_{\mathbf{i}_{2k-2}}t^1_{i_{1}}\prod_{l=1}^{k-2} t^{i_{2l}}_{i_{2l+1}}t^{i_{2k-2}}_{N}\\
	&\quad+\sum_{\mathbf{i}_{2k-2}\in C_{N-4}^{2,2k-2}}(-1)^N\tau_{\mathbf{i}_{2k-2}}t^1_{i_{1}}\prod_{l=1}^{k-2} t^{i_{2l}}_{i_{2l+1}}t^{i_{2k-2}}_{N-1}+\sum_{\mathbf{i}_{2k-2}\in C_{N-4}^{2,2k-2}}\tau_{\mathbf{i}_{2k-2}}(-1)^{N+2}t^2_{i_{1}}\prod_{l=1}^{k-2} t^{i_{2l}}_{i_{2l+1}}t^{i_{2k-2}}_{N}
	\\
	&\quad+\sum_{\mathbf{i}_{2k-1}\in C_{N-4}^{2,2k-1}}(-1)^2\tau_{\mathbf{i}_{2k-1}}t^2_{i_{1}}\prod_{l=1}^{k-1} t^{i_{2l}}_{i_{2l+1}}+\sum_{\mathbf{i}_{2k-1}\in C_{N-4}^{2,2k-1}}(-1)\tau_{\mathbf{i}_{2k-1}}t^1_{i_{1}}\prod_{l=1}^{k-1} t^{i_{2l}}_{i_{2l+1}}\\
	&\quad+\sum_{\mathbf{i}_{2k-1}\in C_{N-4}^{2,2k-1}}(-1)^{N-1}\tau_{\mathbf{i}_{2k-1}}\prod_{l=1}^{k-1} t^{i_{2l-1}}_{i_{2l}}t^{i_{2k-1}}_{N-1}+\sum_{\mathbf{i}_{2k-1}\in C_{N-4}^{2,2k-1}}(-1)^N\tau_{\mathbf{i}_{2k-1}}\prod_{l=1}^{k-1} t^{i_{2l-1}}_{i_{2l}}t^{i_{2k-1}}_{N}
	\\
	&\quad+\sum_{\mathbf{i}_{2k-3}\in C_{N-4}^{2,2k-3}}(-1)^{2N}\tau_{\mathbf{i}_{2k-3}}t^1_{i_{1}}\prod_{l=1}^{k-2} t^{i_{2l}}_{i_{2l+1}}t^{N-1}_{N}+\sum_{\mathbf{i}_{2k-3}\in C_{N-4}^{2,2k-3}}(-1)^{2N+1}\tau_{\mathbf{i}_{2k-3}}t^2_{i_{1}}\prod_{l=1}^{k-2} t^{i_{2l}}_{i_{2l+1}}t^{N-1}_{N}\\
	&\quad+\sum_{\mathbf{i}_{2k-3}\in C_{N-4}^{2,2k-3}}(-1)^{N+3}\tau_{\mathbf{i}_{2k-3}}t^1_{2}\prod_{l=1}^{k-2} t^{i_{2l-1}}_{i_{2l}}t^{i_{2k-1}}_{N}
	+\sum_{\mathbf{i}_{2k-3}\in C_{N-4}^{2,2k-3}}(-1)^{N+2}\tau_{\mathbf{i}_{2k-3}}t^1_{2}\prod_{l=1}^{k-2} t^{i_{2l-1}}_{i_{2l}}t^{i_{2k-1}}_{N-1}
	\\
	&\quad +\sum_{\mathbf{i}_{2k-2}\in C_{N-4}^{2,2k-2}}(-1)^3\tau_{\mathbf{i}_{2k-2}}t^1_2\prod_{l=1}^{k-1} t^{i_{2l-1}}_{i_{2l}}+\sum_{\mathbf{i}_{2k-2}\in C_{N-4}^{2,2k-2}}(-1)^{2N-1}\tau_{\mathbf{i}_{2k-2}}\prod_{l=1}^{k-1} t^{i_{2l-1}}_{i_{2l}}t^{N-1}_N\\
	&\quad +\sum_{\mathbf{i}_{2k-4}\in C_{N-4}^{2,2k-4}}(-1)^{2N+2}\tau_{\mathbf{i}_{2k-4}}t^1_2\prod_{l=1}^{k-2} t^{i_{2l-1}}_{i_{2l}}t^{N-1}_N\\
	& =\sum_{\mathbf{i}_{2k}\in C_{N}^{0,2k}}\tau_{\mathbf{i}_{2k}}\prod_{l=1}^k t^{i_{2l-1}}_{i_{2l}},
	\end{aligned}
	\]
	which competes the proof for the case that $N_0$ is even.
	\item [\textbf{Case ii}] $N_0$ is odd. The proof follows from a similar argument to the case that $N_0$ is even. We shall only briefly sketch it. First, it can be derived that
	\[
	\begin{aligned}
	|\mathbb{P}_N(\beta)|=&(-1)^{N/2}\left(q\(1-t^1_{2}\)\(1-t^{N-1}_{N}\)\left(\sum_{k=0}^{N/2-1} \xi^k q^{N/2-k-1}\mathfrak{A}_k\right)\right.\\
	&
	\left.\quad\quad\quad\quad+q(q-\xi)\(1-t^1_{2}\)t^{N-1}_{N}\left(\sum_{k=0}^{N/2-2} \xi^k q^{N/2-k-2}\mathfrak{B}_k\right)\right.\\
	&
	\left.\quad\quad\quad\quad+q(q-\xi)t^1_{2}\(1-t^{N-1}_{N}\)\left(\sum_{k=0}^{N/2-2} \xi^k q^{N/2-k-2}\mathfrak{C}_k\right)\right.\\
	&
	\left.\quad\quad\quad\quad+(q-\xi)^2t^1_{2}t^{N-1}_{N}\left(\sum_{k=0}^{N/2-2} \xi^k q^{N/2-k-2}\mathfrak{D}_k\right)\right)\\
	=&(-1)^{N/2}\left(\sum_{k=0}^{N/2}\xi^kq^{N/2-k}\mathrm{F}_k\right).
	\end{aligned}
	\]
	Similarly to \eqref{G0}--\eqref{GN2}, we can obtain that $\mathrm{F}_0=1$ and
	\[
	\mathrm{F}_1=\sum_{\mathbf{i}_2\in C_{N}^{0,2}}\tau_{\mathbf{i}_2}t^{i_{1}}_{i_{2}}, \quad \mathrm{F}_{N/2}=t^1_2t^3_4\ldots t^{N-1}_N.
	\]
	We can also deduce that $\mathrm{F}_k$, $k=2, 3, \ldots, N/2-1$, have the same form with \eqref{eq:solpfmn01}, thus there holds
	\[
	\mathrm{F}_k=\sum_{\mathbf{i}_{2k}\in C_{N}^{0,2k}}\tau_{\mathbf{i}_{2k}}\prod_{l=1}^k t^{i_{2l-1}}_{i_{2l}}, \quad k=2, 3, \ldots, N/2-1.
	\]
\end{enumerate}
The proof is complete.
\end{proof}

Next, we proceed with the proof of Theorem \ref{th:ellresult02}.

\begin{proof}[Proof of Theorem \ref{th:ellresult02}]
Arguing by contradiction, suppose that
\begin{equation}\label{eq:contradic}
q^*\not\in \left[-\frac{(\xi-1)^2}{4}, \xi\right], \quad \mbox{or} \quad \Im {q^*}\neq 0.
\end{equation}
Combining \eqref{eq:contradic} with the fact that  $\beta^2-(\xi-1)\beta=q^*$, we have
\[
\beta\not\in [-1, \xi], \quad \mbox{or} \quad \Im {\beta}\neq 0.
\]
Since $\beta$ satisfies the polynomial equation $|\mathbb{P}_N(\beta)|=0$, we can obtain  that there exists non-trivial solution to the following linear equations
\begin{equation}
\mathbb{P}_N(\beta)^T \By=\mathbf{0},
\end{equation}
where $\By\in \RR^N$ and $\By\neq \mathbf{0}$. Next, for $\Bx=(\Bx_1,\Bx_2,\Bx_3)\in \RR^3$, we define
\begin{equation}\label{layerstr01u01}
\mathbf u=\left\{
\begin{aligned}
&a_{N}\Bx, & \Bx&\in D_N,\\
&a_{j}\Bx + b_{j}\frac{\Bx}{|\Bx|^3}, & \Bx&\in D_j,\quad j=N-1, N-2, \ldots , 1\\
&b_{0}\frac{\Bx}{|\Bx|^3}, & \Bx&\in D_{0},
\end{aligned}
\right.
\end{equation}
where $\bm{a}=(a_1, a_2, \ldots, a_N)^T$ and $\bm{b}=(b_0, b_1, \ldots, b_{N-1})^T$ are determined by
\[
\bm{a}=\Xi\By, \quad \bm{b}=\Xi^T\Upsilon_N\By,
\]
where $\Upsilon_N$ is defined in \eqref{eq:matU01} and $\Xi$ is defined by \eqref{eq:defxi01}.
Then it can be verified that $\mathbf u$ defined in \eqref{layerstr01u01} is the solution to
\begin{equation}\label{eq:mainmd01ap01}
\left\{
\begin{array}{ll}
\mathcal{L}_{\tilde{\lambda},\tilde{\mu}} \mathbf{u} =0, & \mbox{in} \quad \RR^3,\\
\mathbf{u}=\Ocal(|\Bx|^{-1}), & |\Bx|\rightarrow \infty,
\end{array}
\right.
\end{equation}
where
\begin{equation}\label{eq:permidf01}
(\tilde{\lambda},\tilde{\mu})= \varepsilon(\Bx)(\lambda,\mu)\quad\mbox{and}\quad
\varepsilon(\Bx)=\sum_{j=0}^{N}\varepsilon_j\chi(D_j)
\end{equation}
and the parameter $\varepsilon_j$ is given by
\begin{equation}\label{eq:vepdef01ap01}
\varepsilon_j
=\left\{
\begin{array}{ll}
\frac{\beta-\xi}{\beta+1}\varepsilon_0, & \mbox{in}\quad D_j, \quad j \quad \mbox{is odd},\\
\varepsilon_0, & \mbox{in}\quad D_j, \quad j \quad \mbox{is even}.
\end{array}
\right.
\end{equation}
Now we claim that $\mathbf u$ is not identically zero vector in $\RR^3$.
Indeed, if  $\mathbf u\equiv\mathbf 0$, it then follows from
\eqref{layerstr01u01} that
\begin{equation}\label{eq:eq0101}
\bm{a}+\Upsilon_N^{-1}M\bm{b}=0, \quad b_0=0,
\end{equation}
where
\begin{equation}\label{eq:defM01}
M=
\begin{bmatrix}
0 & 1 & 0 &\cdots & 0 \\
0 & 0 & 1 &\cdots & 0 \\
\vdots & \vdots &\vdots &\ddots & \vdots\\
0 & 0 & 0 & \cdots & 1\\
0 & 0  & 0 &\cdots & 0
\end{bmatrix}.
\end{equation}
Furthermore, by using the transmission conditions across the boundary of each layer, we have that
\begin{equation}\label{eq:eq0102}
(I-M)\bm{b}=\Upsilon_N (I-M^T)\bm{a}.
\end{equation}
Combining \eqref{eq:eq0101} with \eqref{eq:eq0102}, we obtain that
\begin{equation}
(\Upsilon_N-M\Upsilon_N M^T)\bm{a}=\bm{0},
\end{equation}
which implies that $\bm{a}=\mathbf{0}$. This is a contradiction since $\bm{a}=\Xi\By$ and $\By\neq \mathbf{0}$. Thus the claim is validated.

On the one hand, by using Green's formula \cite{ABGKLW15} for Lam\'e operator, it holds that
\begin{equation}
\begin{aligned}
\int_{\RR^3} \varepsilon\nabla^{s} \mathbf{u}: {\mathbf{C}_{\lambda,\mu}} {\nabla^{s} \mathbf{u}(\mathbf{x})} ~ \mathrm{d} \mathbf{x}&=\sum_{j=0}^N\int_{D_j}\varepsilon_j\nabla^{s} \mathbf{u}: {\mathbf{C}_{\lambda,\mu}} {\nabla^{s} \mathbf{u}(\mathbf{x})} ~ \mathrm{d} \mathbf{x}\\
&=-\varepsilon_0\int_{S_{1}} \left.\frac{\partial \mathbf{u}}{\partial \bm\nu}\right|_{+} \cdot \left.\mathbf{u}\right|_{+} d \sigma(\Bx)+\varepsilon_N\int_{S_{N}} \left.\frac{\partial \mathbf{u}}{\partial \bm\nu}\right|_{-}  \cdot \left.\mathbf{u}\right|_{-}  d \sigma(\Bx)\\
&\quad+\sum_{j=2}^{N}\varepsilon_{j-1}\left(\int_{S_{{j-1}}} \left.\frac{\partial \mathbf{u}}{\partial \bm\nu}\right|_{-} \cdot \left.\mathbf{u}\right|_{-} d \sigma(\Bx)-\int_{S_j} \left.\frac{\partial \mathbf{u}}{\partial \bm\nu}\right|_{+} \cdot \left.\mathbf{u}\right|_{+} d \sigma(\Bx)\right)\\
&=\sum_{j=1}^{N}\left(-\varepsilon_{j-1}\int_{S_j} \left.\frac{\partial \mathbf{u}}{\partial \bm\nu}\right|_{+} \cdot \left.\mathbf{u}\right|_{+} d \sigma(\Bx)+\varepsilon_{j}\int_{S_j} \left.\frac{\partial \mathbf{u}}{\partial \bm\nu}\right|_{-} \cdot \left.\mathbf{u}\right|_{-} d \sigma(\Bx)\right)\\
&=0,
\end{aligned}
\end{equation}
where the last equality follows by using the transmission conditions across the interfaces $S_{j}$, $j=1, 2, \ldots, N$.
On the other hand, by using \eqref{eq:vepdef01ap01} it also holds that
\begin{equation}\label{eq:keyint01}
\begin{aligned}
0&=\int_{\RR^3}\varepsilon\nabla^{s} \mathbf{u}: {\mathbf{C}_{\lambda,\mu}} {\nabla^{s} \mathbf{u}(\mathbf{x})} ~ \mathrm{d} \mathbf{x}\\
&=\varepsilon_0\sum_{j=0}^{\lfloor N/2 \rfloor}\int_{D_{2j}}\nabla^{s} \mathbf{u}(\mathbf{x}): {\mathbf{C}_{\lambda,\mu}} {\nabla^{s} \mathbf{u}(\mathbf{x})} ~ \mathrm{d} \mathbf{x}+ \frac{\beta-\xi}{\beta+1}\varepsilon_0\sum_{j=1}^{\lfloor (N+1)/2 \rfloor}\int_{D_{2j-1}}\nabla^{s} \mathbf{u}(\mathbf{x}): {\mathbf{C}_{\lambda,\mu}} \nabla^{s} \mathbf{u}(\mathbf{x}) ~ \mathrm{d} \mathbf{x}.
\end{aligned}
\end{equation}
We next distinguish the proof into two cases.
\begin{itemize}
	\item [\textbf{Case i}] $\beta\not\in [-1,\xi]$. Then one can readily see that
	\[
	\frac{\beta-\xi}{\beta+1}>0,	
	\]
	which means that the elastic material parameter $\varepsilon_{j},$ $j=1,2,\ldots,N$, are all positive valued. Then by the well-posedness  of the elastostatic system \eqref{eq:mainmd01ap01}, one must have that $\mathbf u\equiv\mathbf0$ in $\RR^3$.
	Indeed,
	in view of \eqref{eq:keyint01}, and using that $\nabla^{s} \mathbf{u}(\mathbf{x}): {\mathbf{C}}_{\lambda,\mu} \nabla^{s} \mathbf{u}(\mathbf{x})\geqslant 0$, we have  that $\nabla^{s} \mathbf{u}(\mathbf{x}): {\mathbf{C}}_{\lambda,\mu} \nabla^{s} \mathbf{u}(\mathbf{x})=0$ in $\RR^3$.
	Thus, from \cite{ABGKLW15}, one then has\[\mathbf{u}\in \bm{\Psi}=\mbox{span}\left\{(1,0,0)^T,\;(0,1,0)^T,\;(0,0,1)^T,\;(\Bx_2,-\Bx_1,0)^T,\;(\Bx_3,0,-\Bx_1)^T,\;(0,\Bx_3,-\Bx_2)^T\right\},\] this, together with the decay behavior of $\mathbf u$ at infinity, implies that $\mathbf u\equiv\mathbf0$ in $\RR^3$.
	\item [\textbf{Case ii}] $\Im {\beta}\neq 0$. It then follows from \eqref{eq:keyint01} that
	\begin{equation}\label{eq:keyint02}
	\Im{\left(\frac{\beta-\xi}{\beta+1}\right)}\sum_{j=1}^{\lfloor (N+1)/2 \rfloor}\int_{D_{2j-1}}\nabla^{s} \mathbf{u}(\mathbf{x}): {\mathbf{C}_{\lambda,\mu}} {\nabla^{s} \mathbf{u}(\mathbf{x})} ~ \mathrm{d} \mathbf{x}=0.
	\end{equation}
	Thus we have
	\begin{equation}\label{oddlayer}
	\nabla^{s} \mathbf{u}(\mathbf{x}): {\mathbf{C}}_{\lambda,\mu} \nabla^{s} \mathbf{u}(\mathbf{x})=0\;\mbox{in}\;D_{2j-1},\; j=1, 2, \ldots, \lfloor (N+1)/2 \rfloor.
	\end{equation}
	By using Green's formula \cite{ABGKLW15} for Lam\'e operator in $D_{2j-1}$ ,$j=1, 2, \ldots, \lfloor (N+1)/2 \rfloor$, and the transmission conditions across the interface $S_{{2j}}$, $j= 1,2, \ldots, \lfloor N/2 \rfloor$, one can deduce that
	\[
	\sum_{j=0}^{\lfloor N/2 \rfloor}\int_{D_{2j}}\nabla^{s} \mathbf{u}(\mathbf{x}): {\mathbf{C}_{\lambda,\mu}} {\nabla^{s} \mathbf{u}(\mathbf{x})} ~ \mathrm{d} \mathbf{x}=0.
	\]
	Thus we also have
	\begin{equation}\label{evenlayer}
	\nabla^{s} \mathbf{u}(\mathbf{x}): {\mathbf{C}}_{\lambda,\mu} \nabla^{s} \mathbf{u}(\mathbf{x})=0\;\mbox{in}\;D_{2j},\; j=1,2, \ldots, \lfloor N/2 \rfloor.
	\end{equation}
	Combining \eqref{oddlayer} with \eqref{evenlayer}, we obtain that $\nabla^{s} \mathbf{u}(\mathbf{x}): {\mathbf{C}}_{\lambda,\mu} \nabla^{s} \mathbf{u}(\mathbf{x})=0$ in $\RR^3$.
	Thus  $\mathbf u\equiv\mathbf0$ in $\RR^3$.
\end{itemize}
We have shown that $\mathbf u\equiv\mathbf0$ holds for either $\beta\not\in [-1,\xi]$ or $\Im {\beta}\neq 0$, which contradicts with our assumptions. The proof is complete.
\end{proof}

\subsection{Extreme case}
In this subsection, we consider that the radius of the layers are extreme large.
It is worth mentioning that the structure of the Earth or other celestial bodies can be treated in this case.
To describe this case mathematically, we set $r_i=R+c_i$, where $R\gg 1$ and $c_i$ are regular constants, $i=1, 2, \ldots, N$. Then we have
\begin{equation}\label{extremecase}
t^i_j=(r_j/r_i)^3=1+\Ocal(1/R).
\end{equation}
Let
\begin{equation}\label{eq:coefg01}
h_{N,k}:=\sum_{\mathbf{i}_{k}\in C_{N}^{0,k}}\tau_{\mathbf{i}_{k}}.
\end{equation}
By direct computations, we have
\[
h_{N,1}=\sum_{\mathbf{i}_{1}\in C_{N}^{0,1}}\tau_{\mathbf{i}_{1}}=((-1)^N -1)/2, \quad h_{N,2}=\sum_{\mathbf{i}_{2}\in C_{N}^{0,2}}\tau_{\mathbf{i}_{2}}=- \lfloor N/2\rfloor,
\]
and
\[
h_{N,N}=\sum_{\mathbf{i}_{2N+1}\in C_{2N+1}^{0,2N+1}}\tau_{\mathbf{i}_{2N+1}}=(-1)^{N(N+1)/2}, \quad h_{N,2\lfloor N/2\rfloor}=\sum_{\mathbf{i}_{2\lfloor N/2\rfloor}\in C_{N}^{0,2\lfloor N/2\rfloor}}\tau_{\mathbf{i}_{2\lfloor N/2\rfloor}}=(-1)^{\lfloor N/2\rfloor}.
\]
Moreover, we can also observe that
\begin{equation}\label{eq:obse01}
h_{N, k}=h_{N-1, k}+(-1)^N h_{N-1, k-1}, \quad k= 2,3, \ldots, N-1,
\end{equation}
and
\begin{equation}\label{eq:obse02}
h_{2N, 2k-1}=\sum_{\mathbf{i}_{2k-1}\in C_{2N}^{0,2k-1}}\tau_{\mathbf{i}_{2k-1}}=0, \quad k=1, 2, \ldots, N.
\end{equation}

Next, we give some recursion formulae for $h_{N,k}$ in the following lemma.
\begin{lem}\label{le:recco01}
Let $h_{N,k}$ be given by \eqref{eq:coefg01}. Then there holds
	\begin{equation}\label{eq:lerec01}
	h_{2N+1,2k}=h_{2N,2k},  \quad k=1, 2, \ldots, N,
	\end{equation}
	and
	\begin{equation}\label{eq:lerec02}
	\begin{aligned}
	h_{2N+2,2k}=h_{2N+1,2k}-h_{2N+1,2k-2}, & \quad k=2, 3, \ldots, N, \\
	h_{2N+1,2k-1}=h_{2N-1,2k-1}-h_{2N-1,2k-3}, & \quad k=2, 3, \ldots, N.
	\end{aligned}
	\end{equation}
\end{lem}
\begin{proof}
First, by using \eqref{eq:obse01} and \eqref{eq:obse02} we have
	\[
    h_{2N+1,2k}=h_{2N,2k}-h_{2N,2k-1}=h_{2N,2k},	
	\]
	which verifies \eqref{eq:lerec01}. Next,  by using again \eqref{eq:obse01} and \eqref{eq:obse02}, we obtain
	\[
	\begin{aligned}
	h_{2N+1,2k-1}&=h_{2N, 2k-1}-h_{2N, 2k-2}\\
	&=-h_{2N, 2k-2}\\
	&=-h_{2N-1, 2k-2}-h_{2N-1, 2k-3}\\
	&=h_{2N-1, 2k-1}-h_{2N-1, 2k-3},
	\end{aligned}
	\]
	where the last equality follows the fact that
	\[
	0=h_{2N,2k-1}=h_{2N-1,2k-1}+h_{2N-1,2k-2}.
	\]
	Then the second equation in \eqref{eq:lerec02} holds. In order to show the first equation in \eqref{eq:lerec02}, we shall make use of induction. It can be readily verified that the first equation in \eqref{eq:lerec02} holds for $N\leqslant 4$. Suppose it holds for $N\leqslant N_0-1$, $N_0\geqslant 5$, we show that it also holds for $N=N_0$.
	Combining \eqref{eq:obse01}--\eqref{eq:lerec01} with the first equation in \eqref{eq:lerec02}, one finally has
	\[
	\begin{aligned}
	h_{2N_0+2,2k}&=h_{2N_0+1, 2k}+h_{2N_0+1, 2k-1}\\
	&=h_{2N_0+1, 2k}+h_{2N_0-1,2k-1}-h_{2N_0-1,2k-3}\\
	&=h_{2N_0+1, 2k}-h_{2N_0-2,2k-2}+h_{2N_0-2, 2k-4}\\
	&=h_{2N_0+1, 2k}-h_{2N_0-1,2k-2}+h_{2N_0-1, 2k-4}\\
	&=h_{2N_0+1, 2k}-h_{2N_0, 2k-2}\\
	&=h_{2N_0+1, 2k}-h_{2N_0+1, 2k-2},
	\end{aligned}
	\]
	which completes the proof.
\end{proof}

In what follows, we begin with the proof of Theorem \ref{co:ellresult02}.

\begin{proof}[proof of Theorem \ref{co:ellresult02}]
By using \eqref{eq:deffq01} and \eqref{extremecase}, the associated polynomial  $f_N(q)$ is then given by
\begin{equation} \label{eq:maincongex02}
f_N(q)=
\sum_{k=0}^{\lfloor N/2\rfloor} \xi^k q^{{\lfloor N/2\rfloor}-k}h_{N,2k}+\Ocal(1/R),
\end{equation}
where $h_{N,2k}$ can be calculated recursively by using Lemma \ref{le:recco01}.
By using elementary combination theory, we can deduce that
\begin{equation}\label{elmcombin}
h_{N,2k}=(-1)^kC_{\lfloor N/2 \rfloor}^k,
\end{equation}
where $C_{\lfloor N/2 \rfloor}^k$ denotes the number of combinations for $k$ out of $\lfloor N/2 \rfloor$. It then follows from  \eqref{eq:maincongex02} and \eqref{elmcombin} that
\begin{equation}
f_N(q)=
\sum_{k=0}^{\lfloor N/2\rfloor} (-1)^k \xi^k q^{{\lfloor N/2\rfloor}-k}C_{\lfloor N/2 \rfloor}^k+\Ocal(1/R)=(q-\xi)^{\lfloor N/2 \rfloor}+\Ocal(1/R).
\end{equation}
The proof is complete.
\end{proof}

\section{Two dimensional case}\label{sec5}
In this section, for the sake of completeness, we consider the elastostatic scattering problem in the two dimensional space $\RR^2$.
\subsection{Representation of the perturbed field}
Let us keep the notations for the multi-layer structure $D$. Suppose the background field $\mathbf{H}$ is represented as
\begin{equation}\label{eq:defH02}
\mathbf{H}=\tilde  a_{0}  \Bx.
\end{equation}
Then the displacement field $\mathbf{u}$ takes the following form
\begin{equation}\label{layerstr02}
\mathbf{u}=\left\{
\begin{aligned}
&\tilde{a}_{N}\Bx, & \Bx&\in D_N,\\
&\tilde{a}_{j}\Bx+r^{-2}\tilde{b}_{j}\Bx, & \Bx&\in D_j,\quad  j=N-1, N-2, \ldots , 1,\\
&\tilde{a}_{0}\Bx+r^{-2}\tilde{b}_{0}\Bx, & \Bx&\in D_{0}.
\end{aligned}
\right.
\end{equation}
By using the transmission conditions across the interface $\{|x|=r_j\}$, $j=1, 2, \ldots N$, we deduce
\begin{equation}\label{eq:trans02}
\left\{
\begin{aligned}
&\tilde{a}_{j}  +\tilde b_{j}r_j^{-2} =\tilde a_{j-1} +\tilde b_{j-1}r_{j}^{-2} , \\
&\varepsilon_{j} \left((2\lambda+2\mu)\tilde a_{j}  - 2\mu\tilde b_{j} r_j^{-2}
\right) = \varepsilon_{j-1} \left((2\lambda+2\mu)\tilde a_{j-1}  - 2\mu\tilde b_{j-1} r_{j}^{-2}
\right)
\end{aligned}
\right.
\end{equation}
where we set $\tilde b_{N}=0$.  By similar argument as in the proof of Theorem \ref{th:solmain01}, we can obtain  the following result:
\begin{thm}
Suppose $\mathbf u$ is the solution to the Lam\'e system
\begin{equation}\label{eq:mainmd02}
\left\{
\begin{array}{ll}
\mathcal{L}_{\tilde{\lambda},\tilde{\mu}} \mathbf{u} =0, & \mbox{in} \quad \RR^2,\\
\mathbf{u}-\mathbf{H}=\Ocal(|\Bx|^{-1}), & |\Bx|\rightarrow \infty,
\end{array}
\right.
\end{equation}
where the Lam\'e parameters $\tilde{\lambda}$ and $\tilde{\mu}$ are given by
\begin{equation}\label{eq:permidf02}
(\tilde{\lambda},\tilde{\mu})= \varepsilon(\Bx)(\lambda,\mu)\quad\mbox{and}\quad
\varepsilon(\Bx)=\sum_{j=0}^{N}\varepsilon_j\chi(D_j).
\end{equation}
 Let $\mathbf H$ be given by \eqref{eq:defH02} and
\begin{equation}\label{eq:deflamb02}
\tilde\beta_{j}=\frac{\frac{\mu}{\lambda+\mu}\varepsilon_{j-1}+\varepsilon_{j}}{\varepsilon_{j-1}-\varepsilon_{j}}, \quad j=1,2,\ldots,N.
\end{equation}
Then we have
\begin{equation}\label{eq:purbmn02}
\mathbf u-\mathbf{H}=r^{-2}\bm{e}^T\tilde \Upsilon_{N} (\tilde{\mathbb{P}}_{N}^T)^{-1}\bm{e}\mathbf{H},
\end{equation}
where $\bm{e}:=(1,1,\ldots,1)^T$, the matrix $\tilde{\mathbb{P}}_N$ and $\tilde\Upsilon_{N}$ are defined by
\begin{equation}\label{eq:matP022}
\tilde{\mathbb{P}}_{N}:= \begin{bmatrix}
\tilde\beta_1 & -1 & -1 & \cdots& -1 \\
\frac{\mu}{\lambda+\mu}(r_{2}/r_1)^2 & \tilde\beta_{2} & -1 &\cdots & -1 \\
\vdots & \vdots & \vdots & \ddots & \vdots\\
\frac{\mu}{\lambda+\mu}(r_{N-1}/r_{1})^2 & \frac{\mu}{\lambda+\mu}(r_{N-1}/r_{2})^2 & \frac{\mu}{\lambda+\mu}(r_{N-1}/r_{3})^2 & \cdots & -1 \\
\frac{\mu}{\lambda+\mu}(r_{N}/r_1)^2 & \frac{\mu}{\lambda+\mu}(r_{N}/r_2)^2 & \frac{\mu}{\lambda+\mu}(r_{N}/r_{3})^2 &\cdots & \tilde\beta_{N}
\end{bmatrix}
\end{equation}
and
\begin{equation}\label{eq:matU02}
\tilde\Upsilon_{N}:= \begin{bmatrix}
r_1^2 & 0 & 0 &\cdots& 0 \\
0 & r_{2}^2 & 0 &\cdots& 0 \\
0 & 0 & r_{3}^2 &\cdots& 0 \\
\vdots & \vdots & \vdots & \ddots & \vdots\\
0 & 0 & 0 &\cdots & r_{N}^2
\end{bmatrix} .
\end{equation}
\end{thm}

\subsection{Polariton resonance modes}
In this subsection, we study the polariton resonances for the elastostatic system in $\RR^2$. As in the 3D case, we shall introduce a similar structural design, i.e., we
assume that $\varepsilon_j$, $j=1,2,\ldots,N$, is given by \eqref{eq:vepdef01}. Then one can readily see that the matrix $\tilde{P}_N$ in (3.6) is reduced to
\begin{equation}\label{eq:matP03}
\tilde{\mathbb{P}}_{N}(\tilde{\beta}):= \begin{bmatrix}
\tilde\beta & -1 & -1 & \cdots& -1 \\
\zeta(r_{2}/r_1)^2 & \zeta-1-\tilde\beta & -1 &\cdots & -1 \\
\vdots & \vdots & \vdots & \ddots & \vdots\\
\zeta(r_{N-1}/r_{1})^2 & \zeta(r_{N-1}/r_{2})^2 & \zeta(r_{N-1}/r_{3})^2 & \cdots & -1 \\
\zeta(r_{N}/r_1)^2 & \zeta(r_{N}/r_2)^2 & \zeta(r_{N}/r_{3})^2 &\cdots & \left(\zeta-1+(-1)^{N}(\zeta-1)\right)/2+(-1)^{N-1}\tilde{\beta}
\end{bmatrix},
\end{equation}
where
\[
\zeta=\frac{\mu}{\lambda+\mu}\quad\mbox{and}\quad\tilde{\beta}=\frac{\zeta\varepsilon_{0}-\varepsilon^*+\mathrm{i}\delta}{\varepsilon_{0}+\varepsilon^*-\mathrm{i}\delta}.
\]
For such a setup, all possible polariton resonance modes are contained in the solution to
\begin{equation}\label{2dcase}
|\tilde{\mathbb{P}}_{N}(\tilde{\beta})|=0,
\end{equation}
which is equivalent to
\[
\tilde \beta^{N-2\lfloor N/2\rfloor}\left(\sum_{k=0}^{\lfloor N/2\rfloor} \zeta^k \left(\tilde\beta^2-(\zeta-1)\tilde\beta\right)^{{\lfloor N/2\rfloor}-k}\left(\sum_{\mathbf{i}_{2k}\in C_{N}^{0,2k}}\tau_{\mathbf{i}_{2k}}\prod_{l=1}^k \left(\frac{r_{2l-1}}{r_{2l}}\right)^2\right)\right)=0.
\]
By similar argument as in the proof of Theorem \ref{th:ellresult02}, we can obtain that
\begin{thm}\label{th:ellresult03}
	Let $\tilde{\mathbb{P}}_{N}(\tilde{\beta})$ be defined in \eqref{eq:matP03} and $\tilde{\beta}^*$ be the root to \eqref{2dcase}.
	Then there exists $N$ real values of roots to \eqref{2dcase}. Moreover, we have
	\begin{equation}\label{eq:thmainsp02}
	\tilde{\beta}^*\in \left[-1, \zeta\right]=\left[-1,\frac{\mu}{\lambda+\mu}\right].
	\end{equation}
\end{thm}
\section{Numerical examples}\label{sec6}
In this section, we present some numerical examples to corroborate our theoretical findings in the previous sections. From a theoretical point of view (see Theorem \ref{th:ellresult01}), in order to find the polariton modes for any multi-layer structure, it suffices to find the roots of the characteristic polynomial \eqref{eq:maincong01}. Moreover, from  Definition \ref{df:def02}, to show the polariton resonance phenomena, one only needs to observe the elastic energy outside the structure, determined
by the term
\[
r^{-3}\bm{e}^T \Upsilon_{N} (\beta I-\mathbb{K}_{N}^T)^{-1} \tilde{\bm{e}}\mathbf H.
\]
 We next investigate these two scenarios numerically.
\subsection{Roots to characteristic polynomials}
First, we consider the radius of layers are equidistance. Let $\varepsilon_0=1$. For $N$-layer structure, set
\begin{equation}\label{eq:str01}
r_i=N-i+1, \quad i=1, 2, \ldots N.
\end{equation}
TABLE \ref{tab:1} shows the roots to the characteristic polynomial  $f_{19}(q)$ with $\xi=3$ and $\xi=0.8$, respectively. In FIGURE \ref{fig:2}, we plot the values of the polynomial $f_{19}(q)$ with $\xi=3$ in the span $[2.59, 3.0]$ (approximately the span between the fourth eigenvalue and last eigenvalue) and the values of the polynomial $f_{19}(q)$ with $\xi=0.8$ in the span $[0.13, 0.8]$ (approximately the span between the first eigenvalue and last eigenvalue), it is interesting that the values of the characteristic polynomial in the span are all very close to zero (of the order $10^{-5}$). These spans may have great potential applications such as wave guide for seismic wave.
\begin{table}[h]
	\begin{center}
		\caption{Roots to the characteristic polynomial with layers which are chosen by \eqref{eq:str01}. \label{tab:1}}
		\begin{tabular}{|c|cccc ccc|}
			\hline
			&  & &  & $N=19$ and $\xi=3$ & &  &  \\
			\hline
			$q$             &   & 2.9691 &  & 2.9496 &  & 2.9185 &  \\
			$\beta$& 0 &2.9923&-0.9923&2.9874&-0.9874&2.9795&-0.9795   \\
			$\varepsilon_i$   & -3.0000 &   -0.0019 &-517.0116 & -0.0032 &-315.5802 &   -0.0051 &-194.3029\\
			\hline
			$q$&& 2.8668 & & 2.7735 & &2.5941 &  \\
			$\beta$ && 2.9664&-0.9664&2.9425&-0.9425&2.8958&-0.8958 \\
			$\varepsilon_i$ &&   -0.0085&-118.0785 &-0.0146  &-68.6177   &-0.0267 &-37.3961 \\
			\hline
			$q$ && 2.2295 & & 1.5099 &  & 0.4949 &  \\
			$\beta$ && 2.7971&-0.7971&2.5843&-0.5843&2.2227&-0.2227\\
			$\varepsilon_i$&&  -0.0534&  -18.7110&   -0.1160&   -8.6217&   -0.2412&   -4.1458\\
			\hline
			&  & &  & $N=19$ and $\xi=0.8$ & &  &  \\
			\hline
			$q$             &   & 0.7918 &  & 0.7866 &  & 0.7783 &  \\
			$\beta$& 0 &   0.7954&  -0.9954&    0.7925 &  -0.9925&    0.7878&   -0.9878   \\
			$\varepsilon_i$   & -0.8000&-0.0026 &-392.1546   &-0.0042 &-239.1326   &-0.0068 &-147.0737\\
			\hline
			$q$&&  0.7645 & & 0.7396 & & 0.6918 &  \\
			$\beta$ && 0.7800&   -0.9800&    0.7658&   -0.9658&    0.7377&   -0.9377 \\
			$\varepsilon_i$ &&   -0.0112 & -89.1786&   -0.0194&  -51.6170&   -0.0358&  -27.9008 \\
			\hline
			$q$ && 0.5945 & &  0.4026 &  & 0.1320 &  \\
			$\beta$ && 0.6775&   -0.8775&    0.5424&   -0.7424&    0.2768&   -0.4768\\
			$\varepsilon_i$&&  -0.0730&  -13.6950&   -0.1670&   -5.9869&   -0.4098&   -2.4404\\
			\hline
		\end{tabular}
	\end{center}
\end{table}
\begin{figure}[!h]
	\centering
	\includegraphics[scale=0.06]{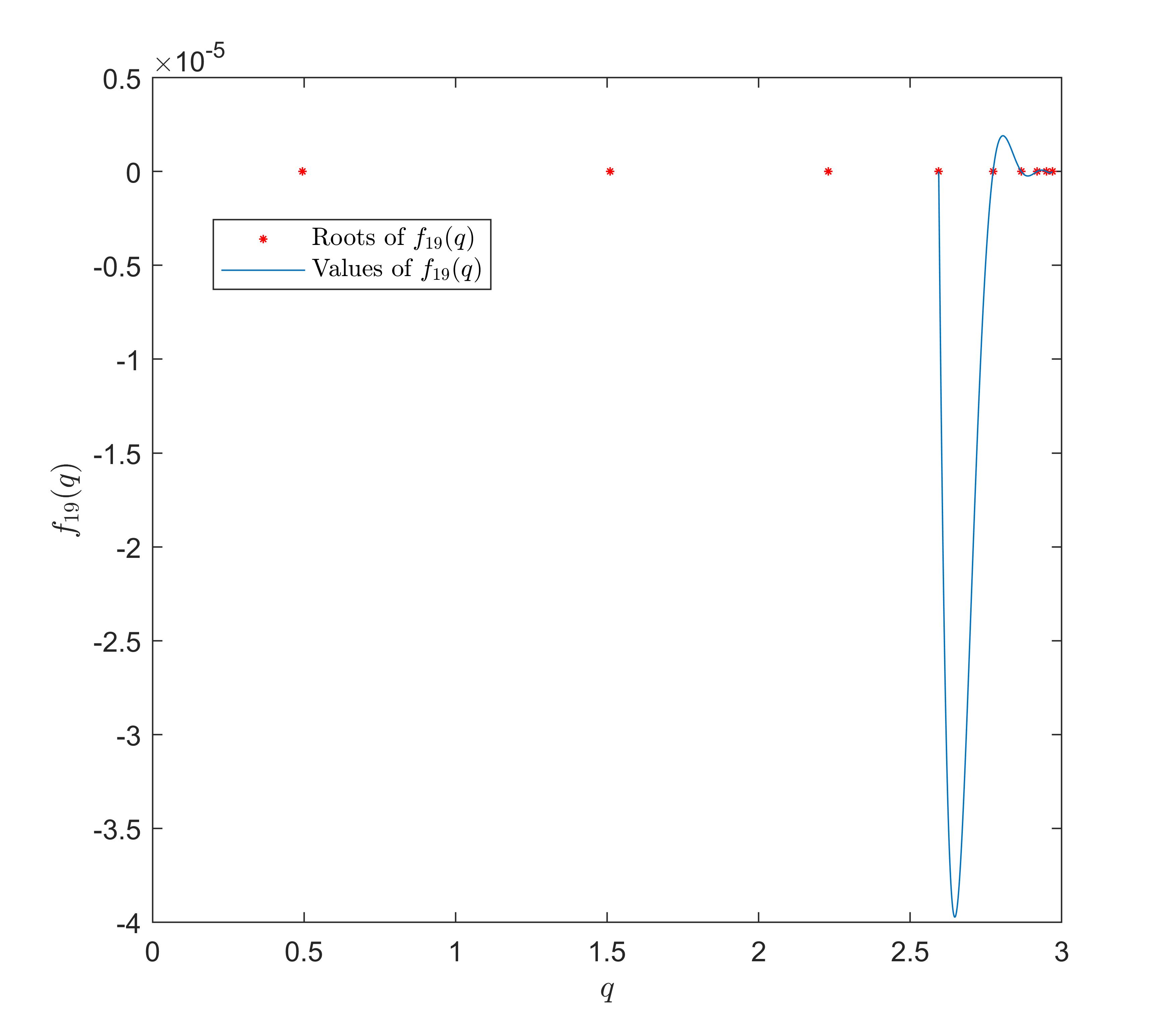}
	\hspace{0.1in}
	\includegraphics[scale=0.059]{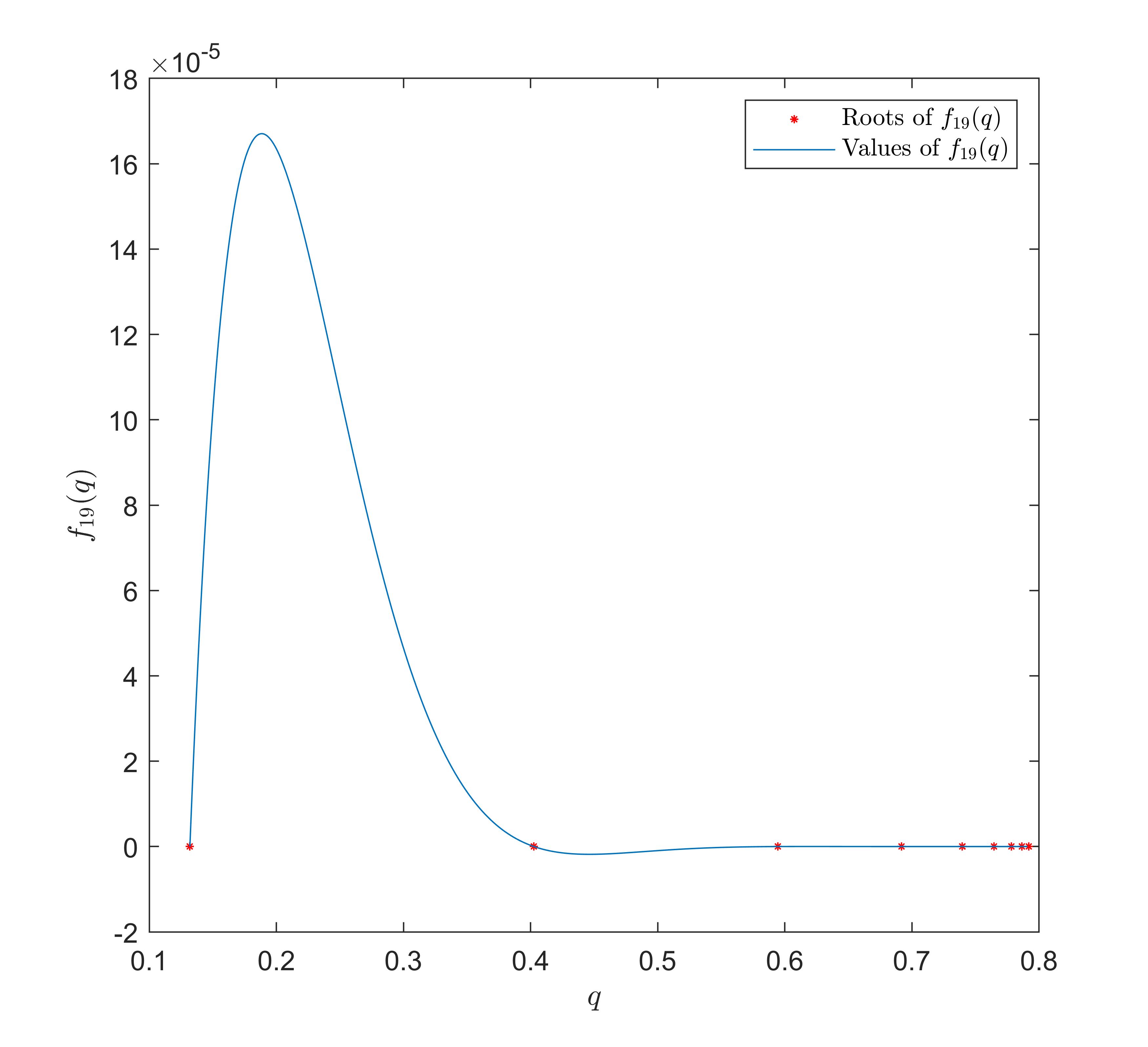}
	\caption{Graph on the left shows the values of $f_{19}(q)$ with $\xi=3$ in the span [2.59,3.0] and plot on
		the right  shows the values of $f_{19}(q)$ with $\xi=0.8$ in the span [0.13, 0.8].}\label{fig:2}
\end{figure}

Next, we consider the radius of layers are decreasing with the same scale $s$, that is
\begin{equation}\label{eq:str02}
r_{i+1}=sr_i, \quad i=1, 2, \ldots N-1.
\end{equation}
Let $r_1=1$ and $s=0.7$. TABLE \ref{tab:222} presents all the roots to the characteristic polynomial  $f_{19}(q)$ with $\xi=3$ and $\xi=0.8$, respectively. Similarly, in FIGURE \ref{fig:3}, one can also find out that the values of the polynomial $f_{19}(q)$ with $\xi=3$ in the span $[1.53, 2.27]$ ( approximately the span between the fourth eigenvalue and last eigenvalue) and the values of the polynomial $f_{19}(q)$ with $\xi=0.8$ in the span $[0.05, 0.61]$ (approximately the span between the first eigenvalue and last eigenvalue),  are all very close to zero. Besides, it is worth mentioning that in both set up of structures, the roots $q$ are all positive values. Similar results can be found in FIGURE \ref{fig:4} and \ref{fig:5} for $N=13$ and $N=16$, respectively.

\begin{table}[h]
	\begin{center}
		\caption{Roots to the characteristic polynomial with layers which are chosen by \eqref{eq:str02}. \label{tab:222}}
		\begin{tabular}{|c|cccc ccc|}
			\hline
			&  & &  & $N=19$  and $\xi=3$ & &  &   \\
			\hline
			$q$             &   & 2.2677 &  & 2.2227 &  & 2.1412 & \\
			$\beta$             & 0 & 2.8077&   -0.8077&    2.7952&   -0.7952&    2.7724&   -0.7724 \\
			$\varepsilon_i$   & -3.0000 &  -0.0505&  -19.7986&   -0.0540&  -18.5314&   -0.0603&  -16.5712 \\
			\hline
			$q$&& 2.0119 & &  1.8171 & &1.5334 &  \\
			$\beta$ && 2.7355&   -0.7355&    2.6784&   -0.6784&    2.5917&   -0.5917 \\
			$\varepsilon_i$ &&   -0.0708&  -14.1215&   -0.0874&  -11.4391&   -0.1137&   -8.7959\\
			\hline
			$q$&& 1.1403 & & 0.6549 &  &0.1979 & \\
			$\beta$ && 2.4630&   -0.4630&    2.2864&   -0.2864&    2.0945&   -0.0945\\
			$\varepsilon_i$ & & -0.1551&   -6.4484&   -0.2171&   -4.6057&   -0.2926&   -3.4174\\
			\hline
			&  & &  & $N=19$  and $\xi=0.8$ & &  &   \\
			\hline
			$q$             &   & 0.6047 &  & 0.5927 &  & 0.5710 & \\
			$\beta$             & 0 & 0.6840&   -0.8840&    0.6764&   -0.8764&    0.6622&   -0.8622 \\
			$\varepsilon_i$   & -0.8000 &  -0.0689&  -14.5228&   -0.0738&  -13.5583&   -0.0829&  -12.0654 \\
			\hline
			$q$&& 0.5365 & &   0.4846 & &0.4089 &  \\
			$\beta$ && 0.6393&   -0.8393&    0.6033&   -0.8033&    0.5472&   -0.7472 \\
			$\varepsilon_i$ &&   -0.0981&  -10.1980&   -0.1227&   -8.1490&   -0.1634&   -6.1211\\
			\hline
			$q$&& 0.3041 & &  0.1746 &  &0.0527 & \\
			$\beta$ && 0.4604&   -0.6604&    0.3297&   -0.5297&    0.1505&   -0.3505\\
			$\varepsilon_i$ & & -0.2325&   -4.3008&   -0.3537&   -2.8274&   -0.5645&   -1.7715\\
			\hline
		\end{tabular}
	\end{center}
\end{table}
\begin{figure}[!h]
	\centering
	\includegraphics[scale=0.06]{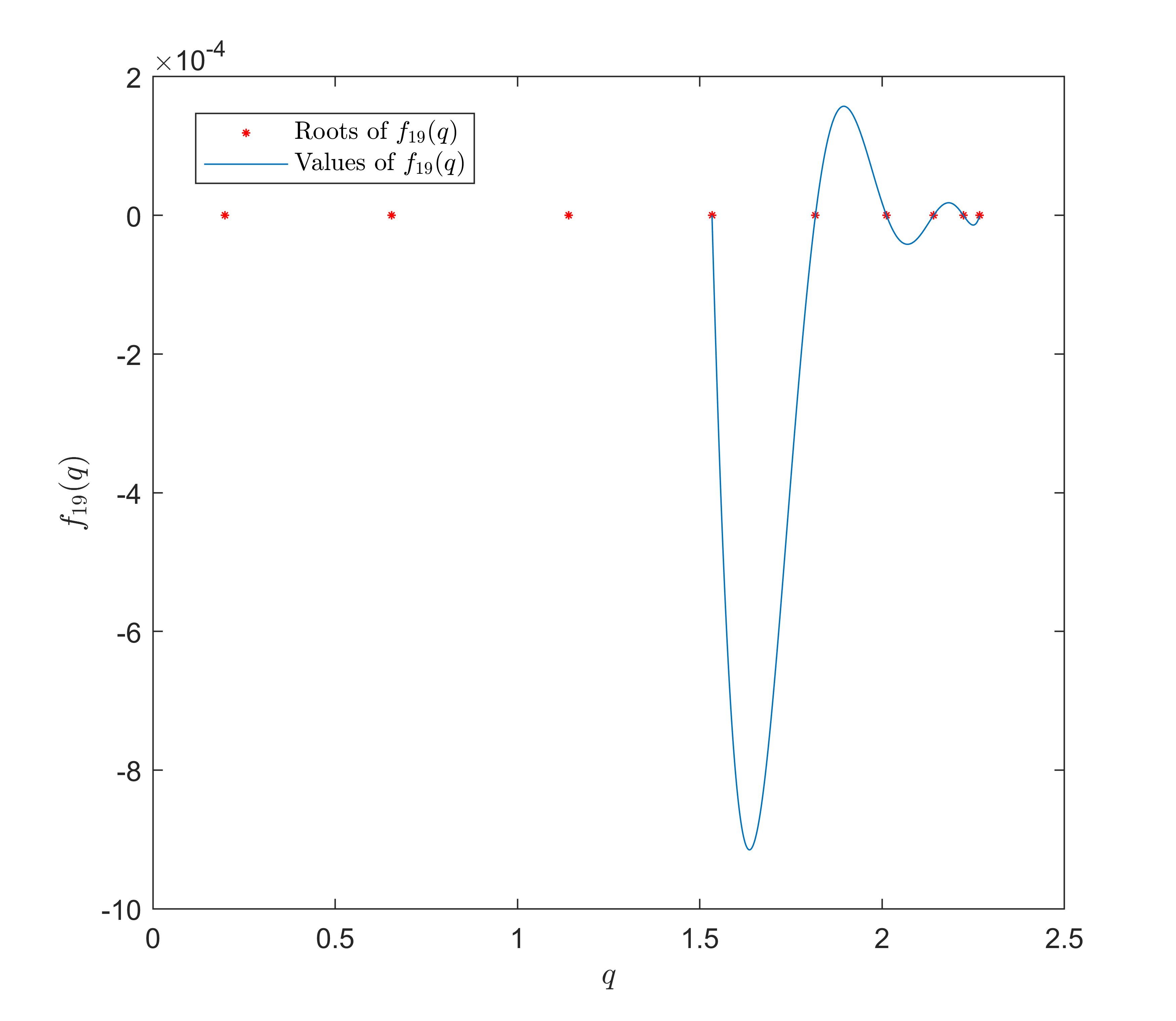}
	\hspace{0.1in}
	\includegraphics[scale=0.06]{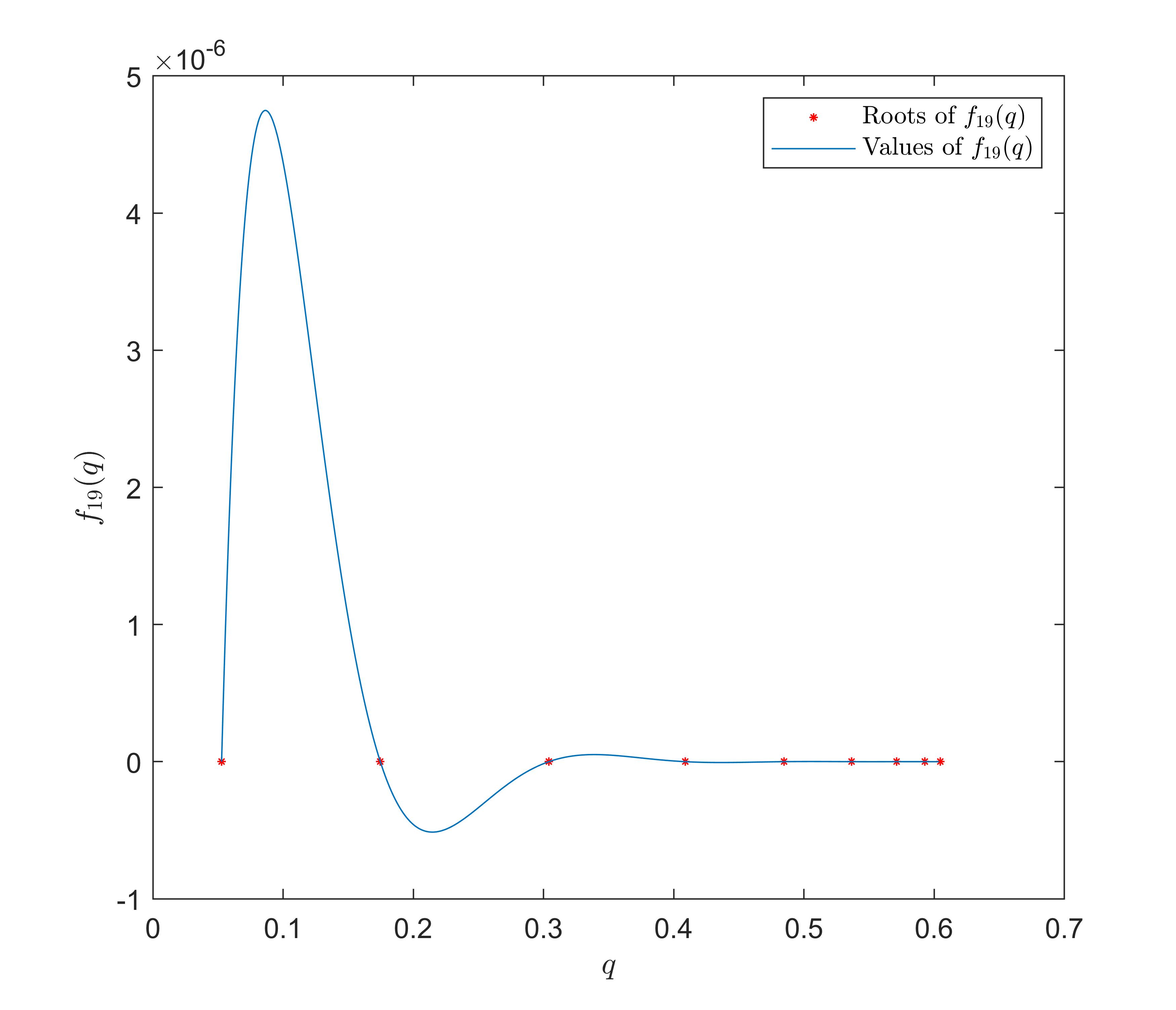}
	\caption{Graph on the left shows the values of $f_{19}(q)$ with $\xi=3$ in the span [1.53,2.27] and plot on
		the right  shows the values of $f_{19}(q)$ with $\xi=0.8$ in the span [0.05, 0.61].}\label{fig:3}
\end{figure}
\begin{figure}[H]
	\begin{center}
		{\includegraphics[width=5in]{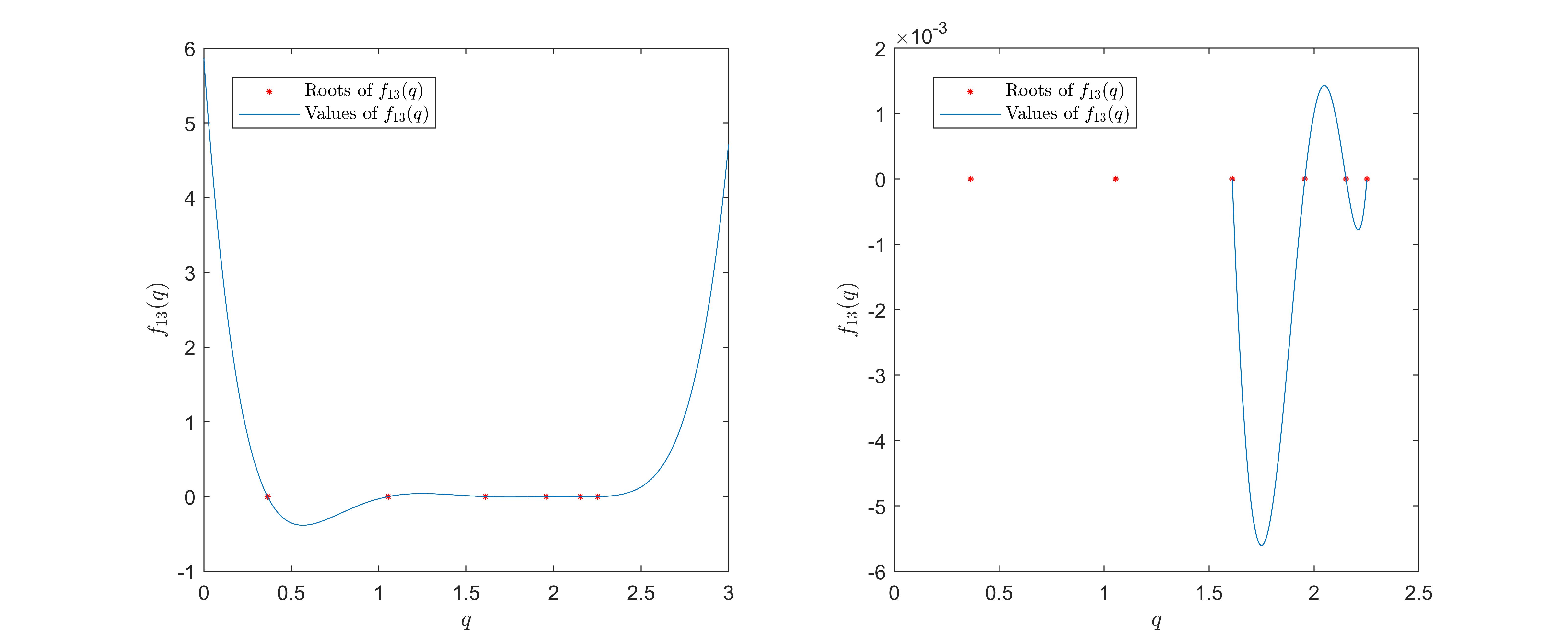}}
	\end{center}
	\begin{center}
		{\includegraphics[width=5in]{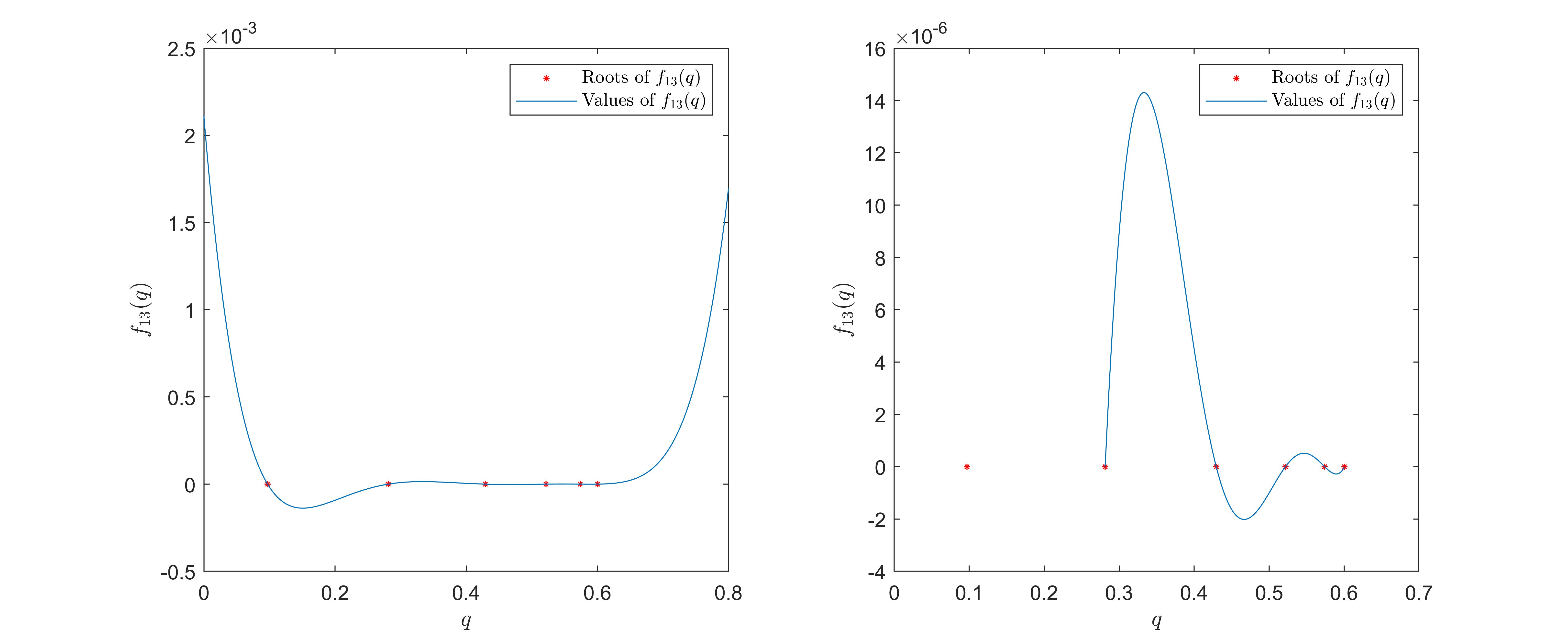}}
	\end{center}
	\caption{Graphs on the first row show the values of $f_{13}(q)$ with $\xi=3$ and plots on
		the second row  show the values of $f_{13}(q)$ with $\xi=0.8$.
	}\label{fig:4}
	\begin{center}
		{\includegraphics[width=5in]{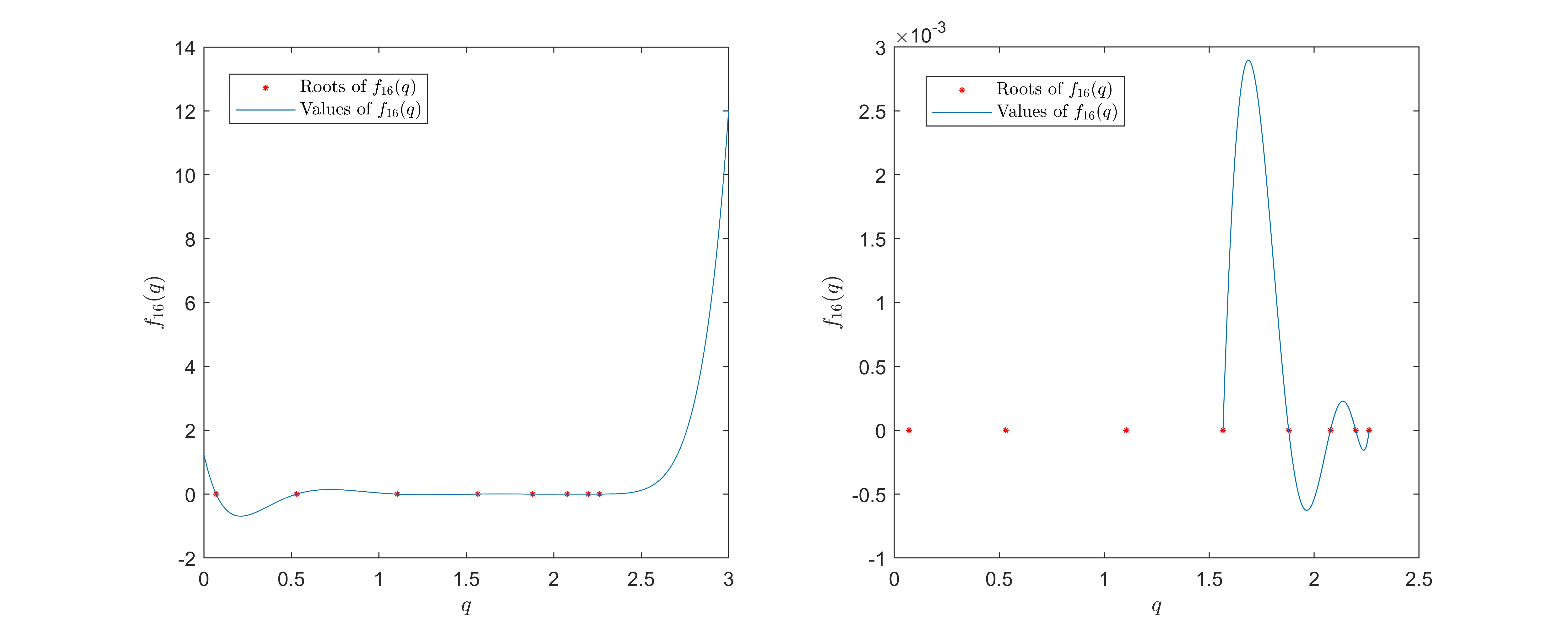}}
	\end{center}
	\begin{center}
		{\includegraphics[width=5in]{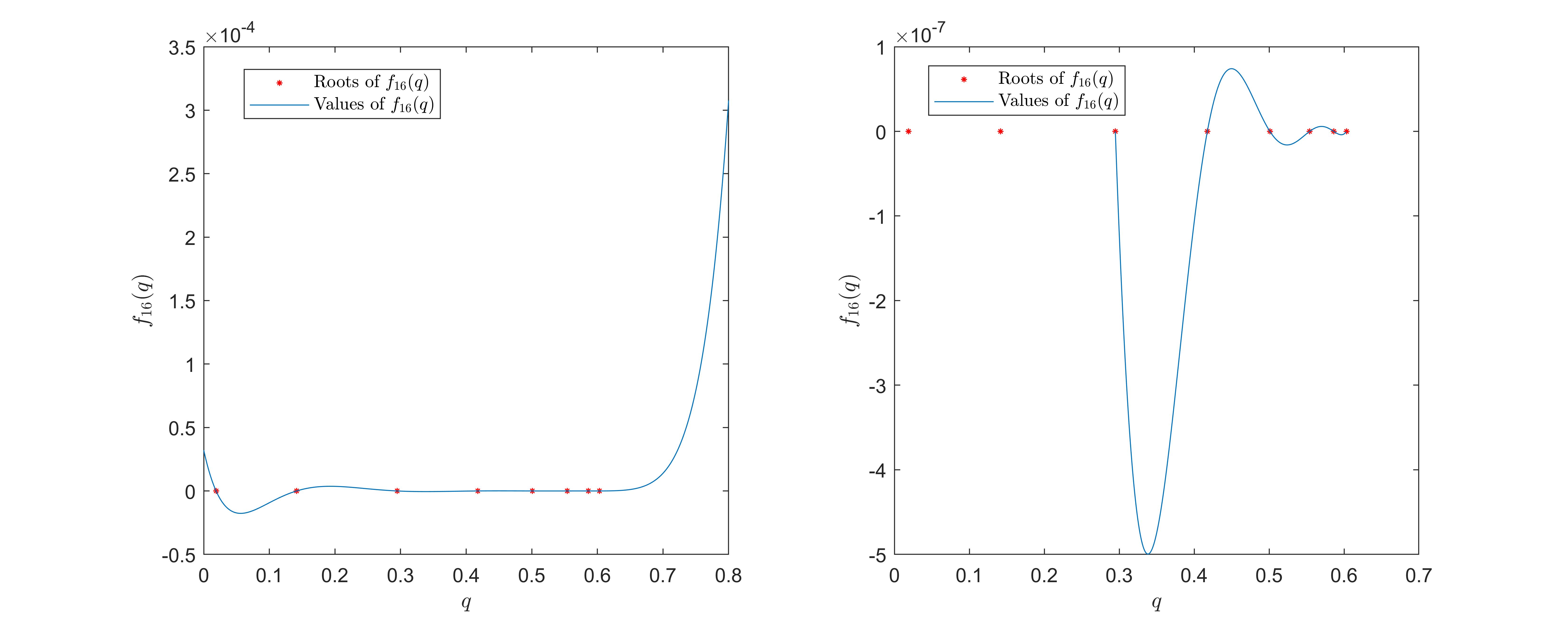}}
	\end{center}
	\caption{Graphs on the first row show the values of $f_{16}(q)$ with $\xi=3$ and plots on
		the second row  show the values of $f_{16}(q)$ with $\xi=0.8$.	
	}\label{fig:5}
\end{figure}
\subsection{Illustration of polariton resonances}
In \eqref{eq:vepdef01}, $\varepsilon_i$ are replaced by
\begin{equation}\label{eq:D}
\varepsilon_i=9 \cdot 10^{-12}\(1-\frac{4 \cdot 10^{30}\delta^2}{1+10^{14}\delta\mathrm{i}}\), \quad i \quad \mbox{is odd}.
\end{equation}
We define the elastic moment tensor here by
\begin{equation}\label{eq:defpol0101}
\mathbf{M}:=r_1^{-3}\Upsilon_{N} (\beta I-\mathbb{K}_{N}^T)^{-1}.
\end{equation}
The functionality of $r_1^{-3}$ appearing in \eqref{eq:defpol0101} is to reduce the scale of the structure.
In FIGURE \ref{fig:6}, we show the norm of the elastic moment tensor defined in \eqref{eq:defpol0101} with the multi-layer structure designed by \eqref{eq:str01}, where $\xi=3$ and $N=14$. It can be seen that the peaks of the norm of the elastic moment tensor are in accordance with the polariton modes in the setup \eqref{eq:D}.
\begin{figure}[!h]
	\begin{center}
		{\includegraphics[width=3in]{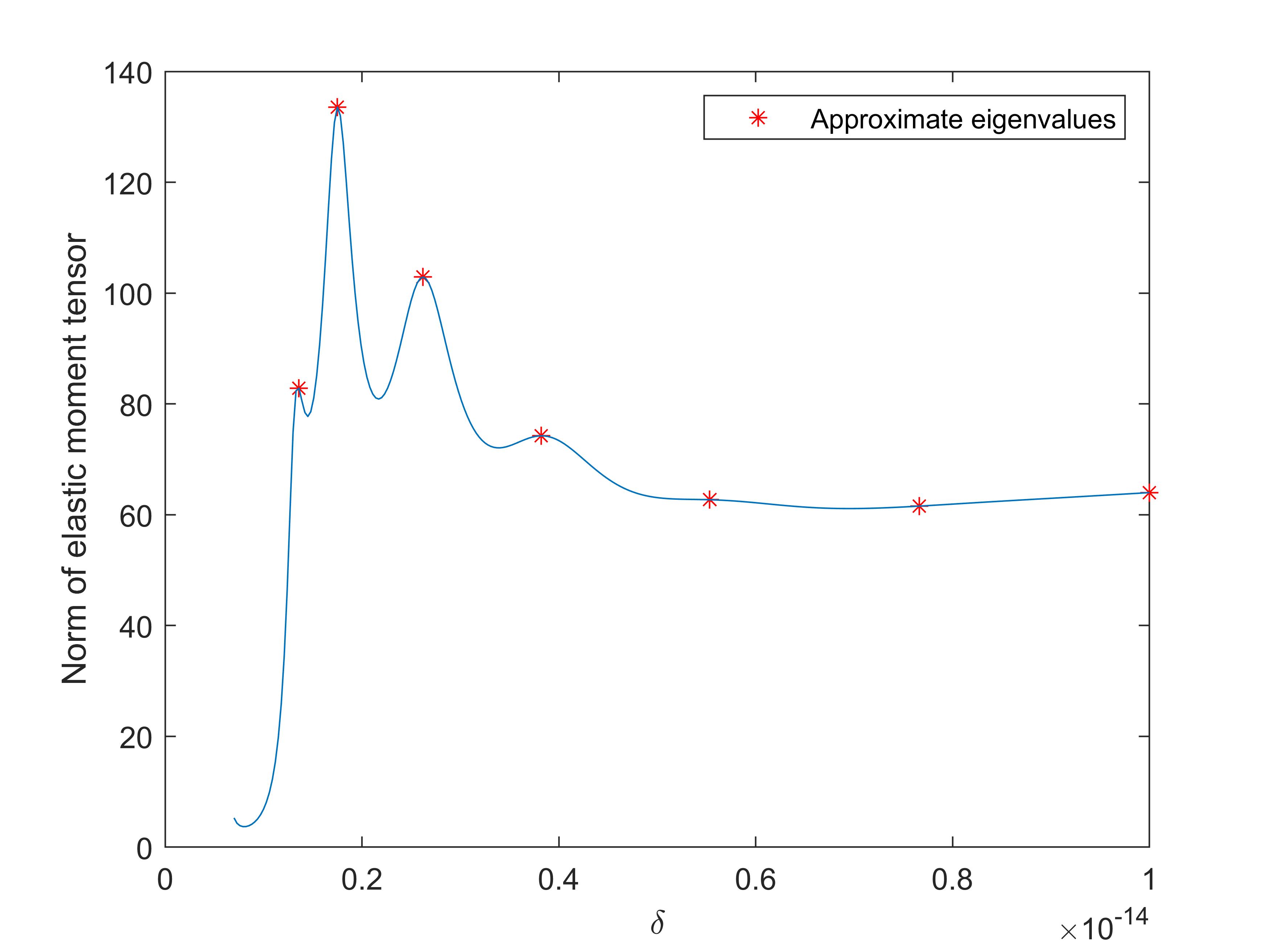}}
	\end{center}
	\caption{Norm of elastic moment tensor  with layers which are chosen by \eqref{eq:str01}.
	}\label{fig:6}
\end{figure}

\section{Concluding remarks}\label{sec7}

In this paper, we studied the elastostatic scattering from a rather general multi-layer metamaterial structures and derived the exact scattering field in terms of the elastic momentum matrix. By highly intricate and delicate analysis of the momentum matrix, we established a handy algebraic framework for analysing the polariton resonances associated with such material structures which yields explicit relationships between the polariton resonance and the geometric and material configurations of the structure. This facilitates the design of metamaterial structures to induce customised resonances. Our study opens up an intriguing direction of further development with many possible extensions as well as applications. In practice, the multi-layer structure can serve as the building block for various material devices. For example, it may serve for constructing elastic/seismic waveguide with enlarged SPR-like (Surface Polariton Resonance) band by employing more layers of metamaterials.  As also discussed in the introduction, it can be used for the quantitative design of ESC-vanishing structures to produce elastic invisibility cloaking devices with enhanced cloaking effects with a rigorous basis. Moreover, our study can be extended to the other physical systems including the electromagnetism and acoustics. We shall study these and other extensions in our forthcoming works.

\section*{Acknowledgement}
The work of Y. Deng was supported by NSFC-RGC Joint Research Grant No. 12161160314 and NSF grant of China No. 11971487.
The work of H. Liu was supported  by NSFC/RGC Joint Research Scheme, N\_CityU101/21, ANR/RGC Joint Research Scheme, A-CityU203/19, and the Hong Kong RGC General Research Funds (projects 12302919, 12301420 and 11300821).

\section*{Data availability statement}
This is a mathematical paper containing all the necessary theoretical proofs. 
There are no data to be reported concerning this work.

\end{document}